\documentclass[11pt,a4paper]{article}

\usepackage[utf8]{inputenc}
\usepackage[T1]{fontenc}
\usepackage{amsmath,amssymb,amsthm}
\usepackage{microtype}
\usepackage{booktabs}
\usepackage{hyperref}
\usepackage{xcolor}
\usepackage{enumitem}
\usepackage{listings}
\usepackage{geometry}
\usepackage{pgfplots}
\pgfplotsset{compat=1.18}

\newcommand{\nmpshape}[3]{\ensuremath{\langle #1, #2, #3 \rangle}}
\newcommand{\nmpfield}[4]{\ensuremath{#1\langle #2, #3, #4 \rangle}}     
\newcommand{\nmpfieldc}[4]{\ensuremath{#1\langle #2, #3, #4 \rangle^{\mathrm{c}}}}
\newcommand{\nmpcount}[5]{\ensuremath{#1\langle #2, #3, #4 \rangle\!:\!#5}} 

\newcommand{\Ff}{\ensuremath{\mathbb{F}}}
\newcommand{\Ftwo}{\ensuremath{\Ff_2}}
\newcommand{\Fthree}{\ensuremath{\Ff_3}}
\newcommand{\Reals}{\ensuremath{\mathbb{R}}}
\newcommand{\Complex}{\ensuremath{\mathbb{C}}}
\newcommand{\Rationals}{\ensuremath{\mathbb{Q}}}
\newcommand{\Integers}{\ensuremath{\mathbb{Z}}}

\newtheorem{theorem}{Theorem}
\newtheorem{lemma}[theorem]{Lemma}

\theoremstyle{definition}

\theoremstyle{remark}

\newcommand{\code}[1]{\texttt{#1}}

\title{A catalog of fast matrix multiplication algorithms\\
with exhaustive derivations\thanks{Source repository:
\href{https://github.com/solven-eu/matmulcatalog}{github.com/solven-eu/matmulcatalog}.
This article is generated from \texttt{paper/article.tex} alongside the
catalog; numerical tables track the catalog by construction.}}

\author{Benoit Lacelle\\
\texttt{benoit@solven.eu}}

\date{Draft, \today}

\begin{document}
\maketitle

\begin{abstract}

The 2022--2026 burst of activity in small-format matrix multiplication
(AlphaTensor 2022, AlphaEvolve 2025, Schwartz--Zwecher 2025, and
Perminov's open-source flip-graph
framework~\cite{perminov2026fast,perminov2025fast}) has produced
striking results but scattered them across fields, attribution
conventions, and serialisation formats.
We present a unified, machine-checkable catalog covering shapes up to
\nmpshape{32}{32}{32} over \Rationals, \Integers, \Reals, \Complex,
\Ftwo, and \Fthree, with a separate axis for commutative algorithms
(Waksman 1970, Makarov 1986, Rosowski 2019). Derivation over the
catalog applies a fixed set of well-defined operators exhaustively ---
axis-flip, Kronecker (with its serendipitous ``bud'' variant), axis
concatenation, recombination-with-allocation (with optional output
peeling and leaf-level pair fusion), and downward projection --- iterated to a
fixed point (Section~\ref{sec:strategies}).
Our derivation layer is closest in spirit to Sedoglavic's FMM-Lille
catalog~\cite{fmmlille}, which likewise layers recursive derivations
(Kronecker with serendipity, concatenation, projection) over known bases; we
differ by covering more fields (\Integers, \Ftwo, \Fthree, and a
commutative axis alongside \Rationals/\Reals/\Complex) and by exploring
the derivation space more exhaustively --- rather than sampling it
stochastically, as the flip-graph / meta-flip-graph
methods~\cite{perminov} do.

As of the 2026-07-12 snapshot over \Rationals, our derivations improve
on the best external catalog (FMM-Lille or Perminov) at 908 shapes ---
\(19\%\) of the 4760 comparable shapes in the large-format band
\(17 \le \max(n,m,p) \le 32\); in fact every strict win we produce
lands in that band, small formats being saturated. We tie the external
best at a further 3688 shapes and trail on only six audited ones, each
explained by a single sharing device (fusion kernels) that our
composition deliberately forgoes.

We refresh the DIS09 comparison tables, split per field with a
commutative column, and provide tooling to regenerate them as the
catalog evolves.

\end{abstract}

\section{Introduction}\label{sec:intro}

The complexity of multiplying two \(n \times n\) matrices has been
studied since Strassen's 1969 algorithm~\cite{strassen1969} broke the cubic barrier, and
remains open: the exact value of the matrix multiplication exponent
\(\omega\) is unknown, though sub-cubic upper bounds have been driven
down to \(\omega < 2.371\ldots\) via the laser method and its
successors. The asymptotic regime has been productive for theoretical
upper bounds but has produced no concrete algorithms one could implement
for finite \(n\); the constants are astronomical.

Meanwhile, the \emph{small-format} regime -- exact rank bounds for
fixed \(\nmpshape{n}{m}{p}\) tensors with \(n, m, p \le 32\) or so --
has lived a quieter life dominated by hand-constructed schemes:
Strassen (1969)~\cite{strassen1969}, Hopcroft--Kerr
(1971)~\cite{hopcroftkerr1971}, Laderman (1976)~\cite{laderman1976},
Pan (1978, 1980)~\cite{pan1978,pan1980}, and Smirnov (1986,
2013)~\cite{smirnov1986,smirnov2013}. These results are concrete
algorithms; they multiply small matrices using a count of bilinear
products that beats the naïve \(nmp\). When applied recursively, each
small-format rank improvement compounds into an \(\omega\)
improvement. The strongest \emph{practical} small-format exponent comes
from Smirnov's \(\nmpshape{3}{3}{6}\) scheme of rank
40~\cite{smirnov2013}: at \(\omega = 3\log 40 / \log 54 \approx
2.7743\) it sits below the exponent of every explicit cubic scheme in
the catalog, making it the best base case a recursive implementation
can realistically use.

The 2022--2026 period has shaken this landscape, with a fresh record
almost every year. In 2022, AlphaTensor \cite{alphatensor2022} found 47
bilinear products for \(\nmpfield{\Ftwo}{4}{4}{4}\). In 2024,
Kaporin~\cite{kaporin2024brent} gave an explicit complex 48 for
\(\nmpfield{\Complex}{4}{4}{4}\) via a semi-analytical solution of the
Brent equations; in 2025, AlphaEvolve \cite{alphaevolve2025}
independently found 48 for the same shape, which
Dumas--Pernet--Sedoglavic \cite{dps2025} rationalised the same year ---
48 with coefficients in \(\{\pm1,\pm\tfrac12,\pm\tfrac14,\pm\tfrac18\}\),
valid over \(\Rationals\) and hence \(\Reals\). Also in 2025,
Schwartz--Zwecher \cite{schwartzzwecher2025} produced a systematic
disjoint-sum construction giving record ranks at several large shapes. The FMM
catalog at Lille \cite{fmmlille} and Perminov's
\code{FastMatrixMultiplication} repository \cite{perminov} both
aggregate schemes from across the literature \emph{and} contribute
their own results, but via different approaches: FMM layers recursive
derivations (Kronecker with serendipity, concatenation, projection) on top of
known bases, while Perminov runs flip-graph and meta-flip-graph search
to discover fresh low-rank schemes (including ternary-integer ones).

\paragraph{The problem.} These results live in incompatible formats
(Maple scripts, NumPy archives, ad-hoc JSON), use inconsistent shape
conventions (sorted vs unsorted \nmpshape{n}{m}{p}, with or without
explicit field tags), and carry inconsistent attribution. A
\(\nmpshape{3}{3}{3}\) entry sourced from AlphaTensor is sometimes
cited as ``the AlphaTensor scheme'' even though rank 23 was established
by Laderman~\cite{laderman1976} decades earlier -- only some of
AlphaTensor's per-format ranks were genuine improvements over prior
SOTA (its \(\nmpfield{\Ftwo}{4}{4}{4}=47\) among them). The same trap
exists in every bulk-import flow we have audited.

We argue that a catalog that is machine-checkable, fully
attributed, and uniformly serialised is now both possible and
necessary. This paper describes the construction of such a catalog,
together with an exhaustive derivation search that automatically discovers
recombinations of catalog entries.

\paragraph{Field discipline as load-bearing.} A bare \nmpshape{n}{m}{p}
is ambiguous: different fields admit different ranks, as the
\(\nmpshape{4}{4}{4}\) landscape illustrates --- 47 over \Ftwo, 48
over \(\Rationals/\Reals/\Complex\), 49 over \Integers. We keep fields
separate and require every rank claim in the catalog (and in
this paper) to carry a field tag; Section \ref{sec:notation} fixes
the conventions. Commutative algorithms (Waksman
1970~\cite{waksman1970}, Makarov 1986~\cite{makarov1986}, Rosowski
2019~\cite{rosowski2019}) form a separate axis: valid for scalar
matrix multiplication but not liftable to recursive multiplication
over non-commutative rings.

\paragraph{What the catalog catalogs.} Each entry encodes either
\begin{enumerate}[label=(\alph*)]
  \item explicit factor matrices \((U, V, W)\) with field tag,
    addition count, and lineage record; or
  \item a \emph{derived bound} -- a rank given by a closed-form
    construction (Waksman, Rosowski, Pan trilinear aggregation) that we
    reconstruct from its lineage record on demand; or
  \item a \emph{cited bound} -- our last resort, for the edge cases
    where we trust a published rank and attribute it to its source but
    do not (yet) hold its explicit factor matrices.
\end{enumerate}
This three-way split lets us include results we cannot yet
materialise as schemes without misrepresenting their status, and lets
us regenerate the derived layer as our formulas improve.

\paragraph{Derivation over the catalog.} Recursive multiplication
over the same catalog is performed by a search that recombines
existing entries via well-defined operators (Section
\ref{sec:strategies}): axis-flip, Kronecker product (including its
serendipitous ``bud'' variant), axis concatenation, recombination with
allocation (with optional output peeling and pair fusion), and
downward projection -- none of them a disjoint-sum / Sch\"onhage
$\tau$-theorem construction (Section \ref{ssec:disjoint-sum}). The
search is iterated until no further rank improvements appear at any
catalog shape -- a \emph{derivation closure} -- and the resulting wins
are materialised lazily after the search has converged.

\paragraph{The non-overlap property.} An observation that we believe is
underdiscussed: our composition operators never \emph{discover} sharing
between bilinear products while materialising a scheme. A materialised
scheme has exactly the rank the search predicted for it --- the sum
\(\sum_k r_{\text{sub-}k}\) of its sub-schemes' ranks for pure
block-substitution (Kronecker, concatenation, recombination), or a
smaller value where a \emph{known} identity (pair fusion, serendipity,
projection) reduces it, counted before materialisation and never found
during it. What never happens is a product turning out, once expanded,
to do double duty across outputs and merge with another for a free rank
saving. That kind of non-obvious sharing --- as in Strassen's
seven-product \(\nmpfield{\Reals}{2}{2}{2}\) scheme, where a product
serves several output cells at once --- is exactly what hand-crafted
schemes are built around (Laderman, Smirnov, AlphaTensor likewise). We
argue that this makes ``composed scheme matches hand-crafted scheme's
rank'' a meaningful co-existence claim, not a re-discovery claim. The same property
predicts where we \emph{lose}: on the handful of shapes an external
catalog still beats us (six against FMM-Lille), every deficit traces to
one cross-slot sharing device -- a \emph{fusion kernel} -- that our
operators forgo by construction, so those losses confirm the boundary
rather than contradict it. Section \ref{sec:nonoverlap} develops both
directions.

\paragraph{Outline.} Section \ref{sec:notation} fixes notation, the
field-discipline conventions, and the paper's glossary. Section
\ref{sec:architecture} describes the catalog's serialisation, lineage
records, and the explicit/cited/derived split. Section
\ref{sec:strategies} catalogs the derivation operators; Section
\ref{sec:nonoverlap} develops the non-overlap property and its
measured boundary (fusion kernels). Sections \ref{sec:multisets}
and~\ref{sec:hk71} present the two theory contributions ---
recombination multisets, and the constructive realisation of the
Hopcroft--Kerr bound. Section \ref{sec:search} describes the
exhaustive derivation search, and Section \ref{sec:balance} the empirical
case for unbalanced splits. Section \ref{sec:tables} gives the
date-stamped head-to-head against FMM-Lille and Perminov, with an
anatomy of the wins and losses. Section \ref{sec:openquestions}
closes with open questions, and Section \ref{sec:aiagent} records the
human/AI division of labour.

\section{Notation, field discipline, and glossary}\label{sec:notation}

A bilinear algorithm computing the product of an \(n \times m\)
matrix \(A\) by an \(m \times p\) matrix \(B\) consists of three
factor matrices \(U \in K^{nm \times r}\), \(V \in K^{mp \times
r}\), \(W \in K^{np \times r}\) over a field (or ring) \(K\), such
that for every \(A, B\):
\[
  (AB)_{i,j} \;=\; \sum_{k=1}^{r}
  W_{(i,j),k}\;\Big(\sum_{a,b} U_{(a,b),k}\, A_{a,b}\Big)
  \Big(\sum_{a',b'} V_{(a',b'),k}\, B_{a',b'}\Big).
\]
The integer \(r\) is the \emph{rank} (a.k.a.\ \emph{multiplication
count}) of the algorithm, and the quantity this paper focuses on.

\paragraph{Shape notation.} A bare \(\nmpshape{n}{m}{p}\) is not
sufficient -- the rank depends on the algebra. We adopt:
\begin{itemize}[itemsep=2pt]
  \item \(\nmpfield{K}{n}{m}{p}\): non-commutative matmul over the
    field \(K\). E.g.\ \(\nmpfield{\Rationals}{3}{3}{3} = 23\)
    (Laderman).
  \item \(\nmpfieldc{K}{n}{m}{p}\) (superscript-\(c\)): commutative
    matmul over \(K\). E.g.\ \(\nmpfieldc{\Rationals}{3}{3}{3} = 21\)
    (Rosowski) --- commutativity saves two multiplications at the same
    shape.
\end{itemize}
The value after \(=\) is always the \emph{rank} --- the number of
scalar multiplications, which is the quantity that compounds into an
\(\omega\) bound under recursion. Addition counts are tracked
separately (Section~\ref{sec:architecture}).

\paragraph{The fields we track.}
\begin{description}[itemsep=2pt]
  \item[\Integers, \Rationals, \Reals] Characteristic-zero and
    non-commutative-friendly (they lift to recursive matmul). We keep
    them \emph{distinct}, first-class fields --- each with its own rank
    claim and attribution, never collapsed into one combined
    \(\Integers/\Rationals/\Reals\) bucket --- and use the inclusion
    \(\Integers\subset\Rationals\subset\Reals\) only to compute which
    fields a scheme satisfies. At \(\nmpshape{4}{4}{4}\), for instance,
    \(\Integers\) stands at 49 (Strassen\(^2\)) while \(\Rationals\) and
    \(\Reals\) reach 48 (Dumas--Pernet--Sedoglavic~\cite{dps2025}) ---
    the same shape, genuinely different ranks. Ternary-integer
    schemes (\(\Integers\)T, coefficients in \(\{-1,0,1\}\)) are a
    first-class sub-class of \(\Integers\).
  \item[\Complex] A complex extension; real-valid schemes work
    immediately. AlphaEvolve's complex-coefficient
    \(\nmpfield{\Complex}{4}{4}{4} = 48\) lives here.
  \item[\Ftwo] Characteristic-2. AlphaTensor's
    \(\nmpfield{\Ftwo}{4}{4}{4} = 47\) lives here. Integer-valid
    schemes reduce mod 2 to give \Ftwo\ schemes; the field stampers
    grant this systematically (the promotion made \({\sim}8000\)
    integer schemes visible under the \Ftwo/\Fthree\ selectors).
  \item[\Fthree] Characteristic-3 (GF(3)) --- ``ternary modular'',
    not to be confused with the ternary-\emph{integer} ZT sub-class of
    \Integers. We track it so the discipline is genuinely
    field-agnostic, though coverage is currently thin: a dedicated
    \Fthree\ search is open (Section \ref{sec:openquestions}).
\end{description}

\paragraph{Commutative axis.} Independent of the field choice, an
algorithm may or may not assume that the entries of \(A\) and \(B\)
commute. Commutative-only schemes (Waksman, Rosowski, Makarov 1986)
do \emph{not} lift to recursive matmul over non-commutative rings:
they save bilinear products by exploiting the freedom \(ab = ba\)
that disappears as soon as \(A\) and \(B\) are themselves matrix
blocks, though they remain valid for scalar matmul. Rosowski 2019
(Theorem~5, \(n\) odd) reaches
\(\nmpfieldc{K}{n}{n}{n} \le n(n+1)(n+2)/2 \approx n^3/2\) --- about
half the \(n^3\) scalar products, but a constant-factor saving that
does not lower \(\omega\) over non-commutative rings. The catalog
flags every commutative scheme with a \code{"commutative": true} tag,
and the comparison tables of Section \ref{sec:tables} carry a
separate commutative column for each field, because
cross-contamination silently produces wrong ``wins'' against
historical non-commutative tables; the canonical example is comparing
\(\nmpfieldc{\Reals}{3}{3}{3} = 21\) (Rosowski, commutative) to
\(\nmpfield{\Reals}{3}{3}{3} = 23\) (Laderman, non-commutative)
as if they were the same metric.

We require every rank claim in the paper to carry both a field tag
and a commutativity tag. The reader can recover the metric from the
shape.

\subsection{Glossary}

The paper uses the following vocabulary consistently; variants from
the cited literature (``base'', ``factorisation'') are noted inline.

\begin{description}[itemsep=4pt, leftmargin=2em]

  \item[MM.] Matrix multiplication. Throughout this paper, an MM
    \emph{scheme} for shape \(\nmpshape{n}{m}{p}\) is a bilinear
    algorithm \((U, V, W)\) computing \(C = A \cdot B\) for
    \(A \in K^{n \times m}, B \in K^{m \times p}, C \in K^{n \times p}\)
    over a field \(K\).

  \item[Atom.] A primitive MM scheme not derived from other schemes
    in our framework --- imported from upstream (Perminov,
    AlphaTensor, AlphaEvolve), produced by a closed-form constructor
    (Waksman, Rosowski, Pan TA), or declared as such (Strassen,
    Laderman). Atoms are the leaves of the lineage DAG; some
    literature calls them \emph{base} schemes, but the same scheme
    can be a base for one search round and a derived result of
    another.

  \item[Lineage.] An expression describing in a constructive,
    replayable way how to build an MM scheme from atoms via the
    derivation operators (recombination, Kronecker, concatenation,
    axis-flip, permutation).

  \item[Stub.] A scheme persisted as \emph{lineage only} --- no explicit
    \((U, V, W)\) factor matrices on disk --- whose concrete matrices
    are reconstructed on demand by replaying the lineage
    (Section~\ref{sec:architecture}). Used to store large derived
    schemes compactly and, during search, as the strategy descriptor
    recorded before a candidate is materialised.

  \item[Allocation.] One way to partition a product of shape
    \(\nmpshape{n}{m}{p}\) into blocks compatible with an atom of
    shape \(\nmpshape{n'}{m'}{p'}\): a per-axis composition
    \((a_1, \ldots, a_{n'})\) with \(\sum a_i = n\) on axis A, and
    analogously on B and C. The atom is applied at the outer level,
    with sub-products carried by the blocks of the resulting
    Cartesian-product decomposition.

  \item[Shape \(\nmpshape{n}{m}{p}\).] A matrix-product shape with
    \(A \in K^{n \times m}, B \in K^{m \times p}\), giving
    \(C \in K^{n \times p}\). We always tag the algebra:
    \(K\nmpshape{n}{m}{p}\) for non-commutative,
    \(K\nmpshape{n}{m}{p}^c\) for commutative; the bare
    \(\nmpshape{n}{m}{p}\) form is reserved for contexts where the
    algebra has just been named.

  \item[Optimisation target.] This catalog does \emph{not yet} target
    minimal additions per scheme: two schemes of equal rank at the
    same shape (Strassen-18 vs Strassen--Winograd-15) induce the same
    recursion structure and the same \(\omega\), so for now we treat
    addition count as metadata rather than a primary quality measure ---
    though we still record it, and flag specific cases of interest such
    as the Strassen-18 vs Strassen--Winograd-15 pair above. What we \emph{do} minimise is
    the multiset of distinct sub-shapes \(\nmpshape{n'}{m'}{p'}\) the
    recursion induces --- fewer distinct shapes means a simpler block
    decomposition, less catalog dependency, and easier verification.
    An illustration from the 2026-06 snapshot: our then-best
    \(\nmpshape{17}{17}{17}{=}2940\) used three distinct sub-shapes
    where FMM-Lille's then-best \(2934\) used two --- the two-sub-shape
    recipe is structurally preferred regardless of addition counts.
    (Both catalogs have since reached \(2930\) at this shape; the
    criterion is unchanged.) One caveat on addition figures: the
    Strassen--Winograd ``15'' is a \emph{scheduled} count (after
    common-subexpression sharing), while flat-counting the \((U, V,
    W)\) forms gives \(\sim\)22--24 --- so raw addition counts are an
    unreliable cross-scheme comparison
    (Section~\ref{ssec:strassen-vs-winograd}).

  \item[Kronecker product.] Composition of a scheme at
    \(\nmpshape{n_1}{m_1}{p_1}\) with a scheme at
    \(\nmpshape{n_2}{m_2}{p_2}\) into a scheme at
    \(\nmpshape{n_1 n_2}{m_1 m_2}{p_1 p_2}\) of rank exactly
    \(r_1 \cdot r_2\) (Section~\ref{ssec:kronecker}).

  \item[Permutation.] Rewrite of a scheme under an \(S_3\)
    permutation of its three axes, yielding a rank-preserving
    six-element orbit of equivalent shapes
    (Section~\ref{sec:strategies}).

  \item[Axis-flip.] Rewrite that reverses the row order on any
    subset of the three axes, yielding a rank-preserving orbit of
    eight variants per scheme (Section~\ref{ssec:axis-flip}).

  \item[Isotropy + coordinate change.] The full symmetry group of a
    bilinear-algorithm tensor: the continuous coordinate changes
    \(GL_n(K) \times GL_m(K) \times GL_p(K)\) together with the discrete
    \(S_3\) axis symmetry; the isotropy group of a decomposition is the
    sub-group that fixes it~\cite{burichenko2014}. The \emph{Permutation}
    and \emph{Axis-flip} rewrites above are the finite \emph{monomial}
    corner of this group --- Permutation is the \(S_3\) relabelling of
    the three axes, axis-flip a row-reversal (a monomial element of the
    \(GL\) factors) --- the rank-preserving, cheaply enumerable
    symmetries we exploit directly, as opposed to the full continuous
    group. de~Groote's uniqueness theorem for
    \(\nmpshape{2}{2}{2}\)~\cite{degroote1978} is treated in
    Section~\ref{sec:multisets}.

  \item[Commutative vs non-commutative.] Defined in the
    commutative-axis paragraph of this section; every rank claim
    carries the tag.

\end{description}

\section{Catalog architecture}\label{sec:architecture}

The catalog's source of truth is a tree of \emph{per-scheme} JSON files
under \code{src/main/resources/schemes/section\(N\)/}, bucketed by the
maximum of the three shape dimensions. Each file is one algorithm ---
one shape, one field, one commutativity tag. We do not store every
scheme the search can build, only the ones worth keeping:
\begin{itemize}[itemsep=2pt]
  \item the \emph{elected} best derived scheme for each shape --- we aim
    for at least one per shape, the current best-in-scope construction;
  \item \emph{atoms} worth cataloguing in their own right --- a low
    rank, a low addition count, historical provenance, or any other
    reason to preserve the specimen;
  \item schemes kept for their \emph{ingredient} value: a serendipity or
    projection scheme may not itself hold the best rank at its shape yet
    be actively consumed as a sub-product by another shape's best
    scheme, so we retain it.
\end{itemize}
On top of this tree sit a few \emph{aggregated} JSON files under
\code{docs/}, \emph{generated} from it: the manifest
\code{docs/catalog.json} (one row per elected scheme, consumed by the
web browser and the comparison tooling), alongside
\code{docs/cited-bounds.json} and
\code{docs/derived-from-cited-bounds.json}.

\subsection{Three layers}

\paragraph{Explicit schemes.} JSON files that carry the full \((U,
V, W)\) factor matrices; every explicit scheme passes a randomised
spot-check at load time. The filename is a pure cosmetic label
(shape, rank, source or derivation note, short content hash): every
property --- field, addition count, commutativity --- is read from
the JSON content, never parsed from the name. The hash, the first
seven hex digits of a content hash of the factor matrices, gives
each scheme a stable identity independent of how it was named
upstream.

\paragraph{Derived bounds.} A separate JSON file
\code{docs/derived-from-cited-bounds.json} records rank claims that our own
constructors can re-derive on demand from formulas (Waksman 1970,
Rosowski 2019 Theorem 2/3, Pan 1980 trilinear aggregation, Hopcroft--Kerr
1971 \(\nmpshape{2}{b}{c}\)). Generated by
\code{GenerateDerivedBounds.java}; rebuilds whenever the underlying
formula changes. This is the layer that grows as we ``move'' bounds
from the cited tier to the derived tier by reproducing the original
construction.

\paragraph{Cited bounds.} The failover tier. A central JSON file
\code{docs/cited-bounds.json} records rank claims whose explicit
factor matrices we do not (yet) hold and which no constructor of ours
re-derives --- so we neither materialise nor re-derive them, but still
trust and attribute the published rank. Each entry has shape, field,
commutativity, rank, addition count (or \code{null}), source
citation, and discovery status. This lets us show the result in
comparison tables without pretending to have the scheme. The cited
layer is generated by \code{GenerateCitedBounds.java}.

\subsection{Lineage records}

The \emph{lineage} of a scheme is the exact sequence of composition
steps that produced it, and it is \textbf{replayable}: a
\code{LineageReplayer} reconstructs the full concrete \((U, V, W)\)
factor matrices from the lineage alone. This is essential --- it lets
the catalog store a large scheme as a lineage-only \emph{stub} (no
factor matrices on disk) and rebuild them on demand, and it makes every
composed scheme reproducible from its building blocks rather than
trusted as opaque numbers. The lineage is a DAG whose leaf nodes are
atoms and whose internal nodes are the derivation operators of
Section~\ref{sec:strategies}, one node type per operator. The lineage
is serialised twice: as a
human-readable function-call form (e.g.\
\code{ConcatRight(KronProduct(strassen, winograd), SAME0)})\, and as
a machine-readable JSON DAG with shared-subtree references.

The lineage is load-bearing for three reasons:
\begin{enumerate}[itemsep=2pt]
  \item \textbf{Independent audit.} Because replay is independent of
    any stored matrices, re-deriving from leaves and comparing catches
    silent corruption of an on-disk \((U, V, W)\).
  \item \textbf{Search-without-materialisation.} The exhaustive derivation
    search (Section \ref{sec:search}) writes a strategy descriptor
    -- effectively a lineage stub -- before deciding whether to
    materialise. This lets the search iterate without paying for
    materialisation.
  \item \textbf{Attribution archaeology.} When a bulk-imported
    catalog re-encodes a scheme originally due to an earlier author,
    the lineage records the importing source while a separate
    attribution field (\code{attribution\_for\_rank}) records the
    earliest source that established the bound. Comparison tables
    pull from the latter.
\end{enumerate}

\subsection{Verification}

Four verification modes run independently of the search loop.
\begin{description}[itemsep=2pt]
  \item[Spot-check.] \code{Verifier.passesRandomMatmulSpotCheck}
    samples random \(A\), \(B\), evaluates the scheme via its
    \((U, V, W)\), and compares against the direct product ---
    \(O(S\,r\,(nm{+}mp{+}np))\) for \(S\) samples. Cheap enough to run
    at scheme load time, after every materialisation, and as a nightly
    catalog-wide sweep; probabilistic, but a few samples over a field
    make a false pass vanishingly unlikely.
  \item[Symbolic.] \code{SymbolicVerifier} checks the trilinear
    identity \emph{exactly} (an algebraic equality, no floating-point
    epsilon) and additionally rejects any coefficient outside the
    declared field. Conclusive, but \(O(r\,nm\,mp\,np)\) --- the full
    tensor, \(O(r\,n^6)\) in the cubic case --- so it is reserved for
    the small-dimension schemes where that is affordable.
  \item[Lineage replay.] \code{LineageReplayer} re-constructs the
    scheme from its lineage record, requiring the same
    \((U, V, W)\). Catches divergence between the on-disk scheme and
    the formula it claims to instantiate.
  \item[Field widening.] \code{FieldWideningSweep} re-tags a scheme
    claimed for a wide field down to the strictest field its
    coefficients inhabit (e.g.\ an \(\Reals\)-tagged scheme with all
    coefficients in \(\Rationals\) is really \(\Rationals\)). Within
    characteristic~0 the coefficients \emph{guarantee} the narrower
    tag --- matmul-validity is a system of integer polynomial
    identities, unchanged by restricting to a subfield --- so this is a
    certain re-tag, not a heuristic.
\end{description}

\section{Derivation strategies}\label{sec:strategies}

A \emph{derivation} strategy takes one or more catalog entries and produces a new bilinear
algorithm at a target shape --- combining them upward (\emph{composition}: Kronecker,
concatenation), spanning a size gap with the mirror pair \emph{decomposition} (tile a large
target into blocks handled by a smaller \emph{outer} scheme --- a small-shape scheme
generating a large one) and its reverse \emph{projection} (restrict a larger scheme down to
a sub-shape), or locally reducing the product set (peel, pair fusion); the strategies below
are what the search and materialiser currently know how to do. Convention: the \emph{outer}
scheme is the one whose bilinear structure is reused, its \emph{sub-schemes}
(\emph{leaves}) the ones computing each outer bilinear product at a smaller shape.

\paragraph{A unifying view, and its ceiling: smoothing \(\omega\) across shapes.}
Almost every strategy below is \emph{compositional}: it synthesises a scheme at a target
shape out of schemes at \emph{other} shapes --- the sub-blocks of a decomposition, the
factors of a Kronecker product, the concatenands of an axis split --- shapes that need not
be anywhere near the target (Kronecker \emph{multiplies} them, landing far from either
factor). Read on the landscape of \(\omega\) over shapes, these strategies \emph{interpolate}
in \(\omega\)-value: the rank they produce at \nmpshape{n}{m}{p} is an arithmetic combination
(sum, product) of the ranks --- hence the \(\omega\) --- of the shapes they consume, so the
target's \(\omega\) is a blend of its inputs'. Hence a structural ceiling: the purely additive or
multiplicative blends --- axis concatenation, the plain Kronecker product, plain block-split
decomposition --- \emph{cannot manufacture an \(\omega\) below all of their inputs}: if
every shape they draw on has higher \(\omega\), the arithmetic combination does too. Pushing below every input needs a mechanism that rewrites the scheme's own product
set rather than gluing fixed sub-schemes together; flip-graph and meta-flip-graph
moves~\cite{perminov2025fast,kauersmoosbauerwood2026} are exactly such a mechanism, which is
why they reach ranks our compositional closure cannot.

Two strategies escape this ceiling. The serendipitous (bud) product
(Section~\ref{ssec:serendipitous}) --- a Kronecker \emph{variant}, so it is this variant and
not the plain Kronecker product that can help --- lands \emph{below} the naïve rank product
\(r_1 r_2\) by borrowing and correcting across the factors; and \emph{projection}
(Section~\ref{ssec:projmargin}), more speculatively, can restrict a parent below the blend
by eliminating dead products. Both inject a correction outside the sum or product of the
consumed ranks. That they \emph{do} drop the rank below any pure blend is not hypothetical:
projection and the serendipitous product each root dozens of our SOTA schemes (39 and 29
respectively, Section~\ref{ssec:win-anatomy}), shapes whose best held rank comes from exactly
this correction --- so at those shapes the \emph{local} \(\omega\) is measurably lower, and
we can point to the case. What stays open is \emph{propagation}, not the local drop: a
lowered \(\omega\) at one shape lowers the asymptotic exponent, or lifts another shape, only
if that shape consumes it, and we have no argument that a local correction propagates the way
a flip-graph base improvement provably does --- a local mechanism of unproven global reach.

Table~\ref{tab:strategies} is the section's map. The rows that do not reduce below
\(\sum_k r_k\) fall under the non-overlap property of Section~\ref{sec:nonoverlap}; the
reducing rows apply only \emph{known} identities at predict time, deliberately kept as
separate constructors so that pure recombination never \emph{discovers} sharing.

\begin{table}[h]
\centering
\caption{Derivation strategies used by the catalog. ``Rank \(<\!\sum\)?'' = whether the
operator can produce a scheme with fewer than \(\sum_k r_{\text{sub-}k}\) bilinear products.
\emph{Predict-time} means the reduction is a \emph{known} identity, counted before
materialisation --- no operator \emph{discovers} new sharing while materialising (the
non-overlap property, Section~\ref{sec:nonoverlap}).}
\label{tab:strategies}
\begin{tabular}{lll}
\toprule
Operator & Source & Rank \(<\!\sum\)? \\
\midrule
Axis-flip orbit & folklore + this paper \S\ref{ssec:axis-flip} & no \\
Kronecker product & Strassen 1969 & no \\
Axis concatenation & folklore & no \\
Recombination with allocation & Strassen-style block split & no\textsuperscript{$\ddagger$} \\
Serendipitous product & Smith 2002; Perminov 2026 (buds) & yes (predict-time) \\
Downward projection & Sedoglavic 2017; margin this paper \S\ref{ssec:projmargin} & yes (predict-time) \\
Output peel & DIS09 / Islam 2009 \(\gamma_5\) & no \\
Pair fusion (Pan TA family) & Pan 1980 / KGP 2026 & yes (paired) \\
Leaf-level pairing & Pan 1980 + this paper \S\ref{ssec:pair-fusion} & yes (paired) \\
Disjoint-sum / $\tau$\textsuperscript{$\dagger$} & Pan 1980 / Sch\"onhage 1981 / SZ 2025 & n/a (imported) \\
\bottomrule
\end{tabular}

\smallskip
\noindent\textsuperscript{$\dagger$}\,Imported as explicit factor matrices, never constructed by our search (Section~\ref{ssec:disjoint-sum}).\\
\noindent\textsuperscript{$\ddagger$}\,Exactly \(\sum_k r_{\text{sub-}k}\) by Theorem~\ref{thm:nonoverlap}, bar an accidental leaf coincidence (probed to zero, Section~\ref{sec:nonoverlap}).
\end{table}

\subsection{Axis concatenation (block-additive composition)}
\label{ssec:concat}

If two schemes share two of the three axes, they concatenate along the third: a
\(\nmpshape{n}{m}{p_1}\) scheme of rank \(r_1\) and a \(\nmpshape{n}{m}{p_2}\) scheme of
rank \(r_2\) compose into a \(\nmpshape{n}{m}{p_1 + p_2}\) scheme of rank \(r_1 + r_2\),
computing the two halves of the result independently; analogously along \(n\) and \(m\).
Sometimes called ``plain split along a dimension''. Trivial but weight-pulling: axis
concatenation is the direct (root) derivation behind \(48\) of the \(908\) shapes where our
derived scheme strictly beats every external catalog over \(\Rationals\) (\(5\%\);
Section~\ref{ssec:win-anatomy}) --- the cheapest route to ``rectangular'' shapes (one large
axis) from cubic catalog entries.

\subsection{Kronecker product (multiplicative composition)}
\label{ssec:kronecker}

Given a scheme \(A\) at \(\nmpshape{n_1}{m_1}{p_1}\) of rank \(r_1\) and a scheme \(B\) at
\(\nmpshape{n_2}{m_2}{p_2}\) of rank \(r_2\), the Kronecker product \(A \otimes B\) is a
scheme at \(\nmpshape{n_1 n_2}{m_1 m_2}{p_1 p_2}\) of rank \(r_1 r_2\) --- the classical
recursive substitution behind Strassen's \(\omega \le \log_2 7\): applied recursively to
\(\nmpshape{2^k}{2^k}{2^k}\) it gives \(7^k\) bilinear products. The search enumerates
2-fold Kronecker chains over the pool; higher-arity chains appear via repeated composition.

\subsubsection{Serendipitous product (bud decomposition)}
\label{ssec:serendipitous}

The serendipitous product is a multiplicative composition that, unlike the Kronecker
product, can land \emph{below} the naïve rank product \(r_1 r_2\). Introduced by
Smith~\cite{smith2002fast} and generalised by Sedoglavic in FMM-Lille, we follow the bud-based
formalisation of Perminov~\cite{perminov2026serendipitous} (Definitions \(2.9\)--\(2.12\),
\(\S 2.6\)). Kauers, Moosbauer and Wood~\cite{kauersmoosbauerwood2026} independently
formalised the same fusion (their Section~3, ``divide less, conquer more''): recursive calls
that share an input, or whose output is used in several positions, are treated as a single
larger multiplication --- our \(U\)- and \(W\)-buds --- yielding
\(\omega(\nmpshape{6}{6}{6}) \le 2.8019\). Their construction is a Sch\"onhage-style
\emph{generalized asymptotic sum inequality} applied to the structure of a \emph{given}
decomposition: for \(\nmpshape{6}{6}{6}\) they analyse the flip-graph-found Moosbauer--Poole
rank-\(153\) decomposition to extract its shareable sub-blocks, rather than searching for
them directly. Finding that shareable structure is the same objective as our bud-richness
selection, but it is extracted per decomposition, not produced by a general constructor;
whether a deterministic constructor (a targeted basis change in the spirit of the de~Groote
solve used for kin-row unification by Schwartz and Zwecher~\cite{schwartzzwecher2025}) could
replace the analysis is open.

The key object is a \emph{bud}: a pair (or group) of rank-one terms sharing the
\emph{same} \(u\) vector (or the same \(v\), or the same \(w\)) up to scaling, since
\(u \otimes v \otimes w\) is invariant under
\(u \mapsto \lambda u,\ w \mapsto \lambda^{-1} w\). Buds are basis-free: the shared vector
may be any linear combination of entries, not a single coordinate. Grouping the \(r\)
terms of a scheme by their buds decomposes it into \emph{elementary matmul tensors}
\[
  \mathcal{T} \;=\; \sum_{i} S_i \, \nmpshape{N_i}{M_i}{P_i}, \qquad
  r \;=\; \sum_{i} S_i \, N_i M_i P_i ,
\]
where a \(U\)-bud of \(k\) terms is an \(\nmpshape{1}{1}{k}\) block, a \(V\)-bud a
\(\nmpshape{k}{1}{1}\), a \(W\)-bud a \(\nmpshape{1}{k}{1}\), combined buds give blocks such
as \(\nmpshape{2}{1}{2}\), and every term in no bud is a trivial \(\nmpshape{1}{1}{1}\).
The serendipitous product with a second scheme \(\nmpshape{n_2}{m_2}{p_2}\) applies it to
each elementary block independently and --- crucially --- realises each enlarged block by
the \emph{best known} scheme for its format rather than by the naïve product:
\[
  \mathcal{T} \otimes \nmpshape{n_2}{m_2}{p_2}
  \;=\; \sum_{i} S_i \, \nmpshape{N_i n_2}{M_i m_2}{P_i p_2}, \qquad
  r_s \;=\; \sum_{i} S_i \cdot R(\nmpshape{N_i n_2}{M_i m_2}{P_i p_2}).
\]
A saving over the Kronecker product arises whenever
\(R(\nmpshape{N_i n_2}{M_i m_2}{P_i p_2}) < N_i M_i P_i \cdot R(\nmpshape{n_2}{m_2}{p_2})\)
for at least one block; for trivial blocks the two coincide, so a bud-free scheme reduces
exactly to the Kronecker product.

\paragraph{Worked example: \(\nmpcount{}{4}{8}{12}{272}\).}
The Hopcroft--Kerr \(\nmpcount{}{2}{3}{4}{20}\) scheme has four \(U\)-buds among its twenty
terms, so it decomposes as \(12\cdot\nmpshape{1}{1}{1} + 4\cdot\nmpshape{1}{1}{2}\) (note
\(12 + 4\cdot 2 = 20\)). Taking the serendipitous product with \(\nmpshape{2}{4}{2}\) and
using \(R(\nmpshape{2}{4}{2})=14\) and \(R(\nmpshape{2}{4}{4})=26\),
\[
  \nmpshape{2}{3}{4}\otimes\nmpshape{2}{4}{2}
  \;=\; 12\cdot\nmpshape{2}{4}{2} + 4\cdot\nmpshape{2}{4}{4}, \qquad
  r_s = 12\cdot 14 + 4\cdot 26 = 272,
\]
a scheme for \(\nmpshape{4}{12}{8}\) (\(\nmpshape{4}{8}{12}\) up to axis permutation),
below the Kronecker value \(20\cdot 14 = 280\). The FMM catalog's~\cite{fmmlille} shorthand
``\(({\nmpshape{2}{3}{4}}{:}20 - 8)\otimes\nmpshape{2}{4}{2} + 4\,\nmpshape{2}{4}{4}\)''
records exactly this: eight of the twenty terms live in buds and become the four
\(\nmpshape{2}{4}{4}\) blocks; the other twelve stay \(\nmpshape{2}{4}{2}\).

\paragraph{Minimal example (and why it is not a free lunch).} The smallest serendipitous
saving is Strassen's --- but the product \emph{consumes} \(R(\nmpshape{2}{2}{2})=7\), it
does not produce it. The base \(\nmpshape{1}{1}{2}\) computes \(a\,b_1\) and \(a\,b_2\),
which share the single \(A\)-form \(a\): one \(U\)-bud of size \(2\), i.e.\ the block
decomposition \(\nmpshape{1}{1}{2} = 1\cdot\nmpshape{1}{1}{2}\). Its serendipitous product
with \(\nmpshape{2}{2}{1}\) (rank \(R(\nmpshape{2}{2}{1})=4\)) sends that bud block to the
enlarged format \(\nmpshape{1\cdot 2}{1\cdot 2}{2\cdot 1} = \nmpshape{2}{2}{2}\), realised
by its best known scheme:
\[
  r_s \;=\; 1\cdot R(\nmpshape{2}{2}{2}) \;=\; 7
  \;<\;
  \underbrace{(1\cdot 1\cdot 2)\cdot R(\nmpshape{2}{2}{1})}_{\text{Kronecker product}}
  \;=\; 2\cdot 4 \;=\; 8 .
\]
The bud fuses the Kronecker product's two independent \(\nmpshape{2}{2}{1}\) blocks (eight
products reusing the same \(A\) entries --- that reuse \emph{is} the bud) into one
\(\nmpshape{2}{2}{2}\) block; the saving \(7<8\) exists only because \(\nmpshape{2}{2}{2}\)
is sub-additive --- only because Strassen supplies \(R(\nmpshape{2}{2}{2})=7\) as the
enlarged-block realisation. Buds expose \emph{which} products can merge; a sub-additive
enlarged block --- an \emph{input} to the product --- is what turns the merge into a saving.

\paragraph{Bud-richness, rank, and which schemes to retain.} The rank-\emph{minimal}
scheme for a format is often the \emph{wrong} base for a serendipitous product. Build
\(\nmpshape{6}{6}{6}\) as an inner \(\nmpshape{3}{3}{3}\) (Laderman \(R=23\); Smirnov
\(R(\nmpshape{3}{3}{6})=40\), note \(40<2\cdot 23\)) times one of two \(\nmpshape{2}{2}{2}\)
bases, under either operator:
\begin{center}
\begin{tabular}{lcc}
\toprule
\(\nmpshape{2}{2}{2}\) base & Kronecker & serendipitous \\
\midrule
Strassen, rank \(7\) --- no buds & \(7\cdot 23 = 161\) & \(161\) (\(=\) Kronecker) \\
na\"ive, rank \(8\) --- four \(U\)-buds & \(8\cdot 23 = 184\) & \(4\cdot 40 = \mathbf{160}\) \\
\bottomrule
\end{tabular}
\end{center}
Strassen's \(u,v,w\) are pairwise non-proportional, so it has no buds and serendipity buys
it nothing over Kronecker. The na\"ive rank-\(8\) scheme is bud-rich --- each \(A\)-entry
feeds two outputs, four \(U\)-buds \(4\cdot\nmpshape{1}{1}{2}\) --- and fusing each with the
inner \(\nmpshape{3}{3}{3}\) into a \(\nmpshape{3}{3}{6}\) gives \(4\cdot 40 = 160\): the
best of the four, from the \emph{worse}-rank base. A serendipitous base is chosen for
bud-richness, not rank. Since \(r_s\) depends only on the bud \emph{multiset}
\(\{S_i\,\nmpshape{N_i}{M_i}{P_i}\}\), what must be retained per format is the
\emph{Pareto frontier} in \((\text{rank},\,\text{bud multiset})\): a higher-rank scheme is
worth keeping iff its richer buds win some serendipitous product the lower-rank scheme
cannot; a scheme of no-better rank whose bud multiset is contained in another's is
dominated. Keeping only the rank-minimal scheme per format would discard the winning base.

\paragraph{How operations act on buds.} Bud structure is not preserved uniformly, which is
why it is tracked explicitly. \emph{Rank reduction destroys buds}: it replaces un-mixed
products by independent linear combinations, making \(u,v,w\) pairwise non-proportional ---
empirically every rank-minimal small-format scheme we hold is bud-free. \emph{Kronecker
preserves and multiplies buds} (a product term is proportional to another iff both factors
are), so each bud lifts to the product. \emph{Axis permutation (\(S_3\)) permutes the bud
type}, count unchanged. \emph{Projection can only shrink buds}, while concatenating a
scheme with a copy of itself \emph{creates} \(U\)-buds (a twin with the same \(u\)).

\paragraph{Aligning buds under isotropy.} The bud structure of
\(\mathcal{T}_1\otimes\mathcal{T}_2\) depends on how the factors' bud axes line up: a
size-\(a\) \(U\)-bud aligned with a size-\(b\) \(U\)-bud yields a size-\(ab\) \(U\)-bud; a
crossed pair (\(U\) against \(V\)) reinforces nothing. A basis change preserves
proportionality, so the continuous \(GL\) part of the isotropy group leaves the bud
partition \emph{invariant}; only the discrete \(S_3\) permutation moves buds between axes.
Hence two levers: orienting the bud-provider \(\mathcal{T}_1\) by \(S_3\) selects which
target axis carries the \(k\)-fold enlargement (a \(U\)-bud gives
\(\nmpshape{n_2}{m_2}{k p_2}\), a \(V\)-bud \(\nmpshape{k n_2}{m_2}{p_2}\), whose \(R\)
differ), so the search tries all three orientations; and aligning both factors' bud axes
maximises the \(ab\) reinforcement for reuse as a base.

\paragraph{What it is.} The serendipitous product is \emph{exact} and
\emph{combinatorial}: it reads buds off the base scheme's factor matrices and looks up
sub-block ranks --- no numerical optimisation, no border rank. Our bud recognizer follows
Perminov's \(\S 2.6\), and the assembled scheme is certified by the exact symbolic verifier.

\paragraph{Smoothness governs the serendipitous hit rate.} A serendipitous product needs
the target to factor as \(\nmpshape{n_1}{m_1}{p_1}\otimes\nmpshape{n_2}{m_2}{p_2}\) with
\emph{both} factors non-trivial, and \(\nmpshape{N}{M}{P}\) has \(d(N)\,d(M)\,d(P)\) such
factorisations (per-axis divisor counts). A \emph{smooth} format (many prime factors per
side, e.g. \(30=2\cdot3\cdot5\)) exposes far more base/inner splits than a prime-power
format, each an independent chance for a bud-rich base to meet a sub-additive enlarged
inner. Writing \(\Omega(NMP)\) for the number of prime factors of \(N\,M\,P\) with
multiplicity, Table~\ref{tab:serendip-smooth} reports, over all \(5456\) non-commutative
formats in the catalogue, how often the \emph{best} scheme we hold is serendipitous.

\begin{table}[t]
\centering\small
\begin{tabular}{rrrr}
\toprule
\(\Omega(NMP)\) & formats & serendipity best & \% \\
\midrule
\(\le 5\) & 2013 & 0 & 0.0 \\
6 & 1122 & 6 & 0.5 \\
7 & 979 & 12 & 1.2 \\
8 & 654 & 11 & 1.7 \\
9 & 389 & 9 & 2.3 \\
10 & 178 & 3 & 1.7 \\
\(\ge 11\) & 121 & 0 & 0.0 \\
\midrule
all & 5456 & 41 & 0.75 \\
\bottomrule
\end{tabular}
\caption{How often the best known scheme for a format is a serendipitous product, bucketed
by the smoothness \(\Omega(NMP)\). Winning formats have mean \(\Omega = 7.78\), against
\(6.29\) for the whole population.}
\label{tab:serendip-smooth}
\end{table}

Below \(\Omega=6\) no serendipitous scheme is ever best: the smallest non-trivial
base/inner pair already needs \(\Omega\ge6\) (e.g.
\(\nmpshape{2}{3}{7}\otimes\nmpshape{3}{3}{3}\)). Above the floor the win rate tracks
smoothness, peaking near \(\Omega=9\); the fall-off at \(\Omega\ge11\) is a coverage
artefact (those formats lie beyond our materialisation cap), not a ceiling on the method.
This concerns the practical per-format implied exponent, not the asymptotic \(\omega\);
operationally: serendipitous effort goes to smooth formats first, and prime cubes are
direct-construction targets where the method cannot help.

\subsection{Decomposition: block-splitting and the sub-problem multiset}
\label{ssec:decomposition}

Decomposition, like the composition operators above, builds \emph{up} in shape: it
generates a larger target from smaller schemes (a better small scheme lifts the target ---
an upward edge in the impact DAG of Section~\ref{sec:search}). What differs is only the
direction of \emph{construction} --- instead of gluing two given schemes, it starts from the
target shape and tiles it with a smaller \emph{outer} scheme. Only \emph{projection}
(Section~\ref{ssec:projmargin}) builds \emph{down}, restricting a larger scheme to a
sub-shape.
\emph{Decomposition} splits a large \(\nmpshape{N}{M}{P}\) product into a grid of blocks
computed by a smaller \emph{outer} scheme of shape \(\nmpshape{n_o}{m_o}{p_o}\) and rank
\(r_o\): cut each target axis into that many parts and let each of the outer scheme's
\(r_o\) bilinear products compute one block product of the resulting \(n_o\times m_o\) by
\(m_o\times p_o\) grid --- a \(\nmpshape{2}{2}{2}\) outer scheme turns the target into
\emph{seven} sub-problems, Laderman's \(\nmpshape{3}{3}{3}\) into twenty-three. Listing
those sub-problems by shape, with multiplicity, gives the
\emph{sub-problem multiset} of the decomposition; each is computed independently, so the
realised rank is just the sum of their ranks,
\(R(\nmpshape{N}{M}{P}) \le \sum_{\sigma\in\text{multiset}} R(\sigma)\), and the multiset
is the \emph{only} thing that fixes the cost --- order and labelling are irrelevant. With
all parts equal the sub-problems are identical and the sum collapses to
\(r_o\cdot R(\text{block})\), exactly the Kronecker product (Section~\ref{ssec:kronecker});
the interesting regime is \emph{unbalanced} parts, where the sub-problems differ.

\paragraph{Padding, zero-sparing, and the max.} The realised size of each sub-product is
where padding enters --- the heart of the matter. Split each axis in two with the larger
part first, \(M = M_1 + M_2\), \(N = N_1 + N_2\), \(P = P_1 + P_2\). Each outer product is
a (sum of \(A\)-blocks)\(\,\cdot\,\)(sum of \(B\)-blocks); adding blocks of unequal shape
embeds the smaller into the larger envelope with zeros, and \emph{those zeros are never
multiplied}: a product is realised at its genuine nonzero overlap --- the \emph{max} of the
real parts on each output axis, the \emph{min} on the contraction axis. A product that
sums across an axis is charged that axis's \emph{large} part; a product using a
\emph{single} block stays at the part it touches. Hence the expensive corner
\(\nmpshape{M_1}{N_1}{P_1}\) is paid by every all-sum product; the cheap all-small corner
\(\nmpshape{M_2}{N_2}{P_2}\) only by a product
single-block on \emph{every} axis. This padding-max --- not the rank, not the additions ---
is what will separate Strassen from Winograd on unbalanced splits
(Sections~\ref{ssec:strassen-vs-winograd} and~\ref{ssec:axis-flip}); \textbf{which
template} and \textbf{which orientation} are the two knobs recombination with allocation
(Section~\ref{ssec:recombination}) searches over.

\subsubsection{Strassen, Winograd, and the structural difference}
\label{ssec:strassen-vs-winograd}

The two canonical rank-7 algorithms for \(\nmpfield{\Reals}{2}{2}{2}\) are Strassen 1969
and Winograd's 1971 reformulation. The punchline, stated once: used as a
\(\nmpshape{2}{2}{2}\) template over a split, at a fixed orientation they induce the
\emph{same} multiset of sub-problems and differ only in additions (Strassen 1969: 18;
Winograd 1971: 15); it is the axis-flip orbit of Section~\ref{ssec:axis-flip},
re-orienting the template, that changes the multiplications --- and where the two diverge.
The seven products of each, sized under zero-sparing, are tabulated in
Appendix~\ref{app:worked-products}.

\paragraph{The generic multiset.}
Collecting the realised sizes of either scheme's seven products gives the same multiset:
\[
  2\,\nmpshape{M_1}{N_1}{P_1} + \nmpshape{M_2}{N_1}{P_1} + \nmpshape{M_1}{N_1}{P_2}
  + \nmpshape{M_1}{N_2}{P_1} + \nmpshape{M_1}{N_2}{P_2} + \nmpshape{M_2}{N_2}{P_1}.
\]
The largest corner \(\nmpshape{M_1}{N_1}{P_1}\) appears twice, the all-small corner
\(\nmpshape{M_2}{N_2}{P_2}\) not at all, and the split is \emph{not} symmetric: the
contraction axis \(N\) carries its smaller part in three of the seven products, \(M\) and
\(P\) in two each -- so one minimises the small-part products by making \(N\) the largest
axis. With equal parts (even dimensions) the multiset collapses to
\(7\,\nmpshape{M/2}{N/2}{P/2}\), the textbook recursion.

\paragraph{Worked example: \(\nmpshape{3}{3}{3}\).}
Split \(3 = 2{+}1\) on every axis (\(M_1{=}N_1{=}P_1{=}2\), \(M_2{=}N_2{=}P_2{=}1\)). Both
schemes yield
\[
  2\,\nmpshape{2}{2}{2} + \nmpshape{1}{2}{2} + \nmpshape{2}{2}{1} + \nmpshape{2}{1}{2}
  + \nmpshape{2}{1}{1} + \nmpshape{1}{1}{2},
\]
i.e.\ \(2\cdot 7 + 4 + 4 + 4 + 2 + 2 = 30\) leaf multiplications (recursing each
\(\nmpshape{2}{2}{2}\) with the same rank-7 scheme) -- against \(7\cdot 7 = 49\) for the
naive pad-to-even \(\nmpshape{4}{4}{4}\) embedding and \(3^3 = 27\) schoolbook. The
\(49 \to 30\) gap is the zero-sparing; the further drop from re-orienting the template is
deferred to Section~\ref{ssec:axis-flip}.

\paragraph{Why one needs a \emph{set} of bases.} Strassen and Winograd are only two points
in the space of rank-7 \(\nmpshape{2}{2}{2}\) schemes, and what matters at an unbalanced
split is how many single-block products a scheme has --- hence which sub-problem multisets
its orbit can reach. The pool therefore carries a \emph{set} of rank-7
\(\nmpshape{2}{2}{2}\) bases inequivalent in multiset behaviour, not
just the two textbook ones: we enumerate the rank-7 normal forms (Heun 1994; de~Groote's
classification) exhaustively, one representative per axis-flip orbit. The same idea applies
at any base shape, but the space of rank-\(r\) schemes at \(\nmpshape{3}{3}{3}\) and up is
far too large to enumerate, so there we fall back on the imported/known schemes plus the
axis-flip orbit of each.

\subsubsection{Axis-flip orbit}
\label{ssec:axis-flip}

Reversing the row order of any factor matrix \((U, V, W)\) preserves shape and rank; the
three independent flips generate an orbit of size 8 per scheme (the elementary abelian
\(\mathbb{Z}_2^3\); see the glossary). Under uniform allocations the
flip is a trivial relabeling and the orbit collapses; under \emph{unbalanced} allocations
such as \((9, 8)^3\) the orbit members induce different padding patterns and hence
different total ranks, since flipping the outer scheme permutes which block-index of the
target receives the padding penalty (Section~\ref{ssec:recombination}). The pool stores
one representative per orbit; the search re-enumerates the 8 masks analytically
(per-product block-support bitmasks; Section~\ref{sec:search}) at constant-factor cost.

\paragraph{Why the orbit separates Strassen from Winograd.} This is where the two rank-7
templates stop being interchangeable. Strassen's seven products are all two-block sums, so
under \emph{any} orbit member each still spans a full part on every axis: Strassen keeps
its two full \(\nmpshape{M_1}{N_1}{P_1}\) products and can never realise the all-small
corner. Winograd's two single-block products (\(A_{11}B_{11}\), \(A_{12}B_{21}\))
\emph{can} be flipped onto the small blocks -- one becomes
\(A_{22}B_{22} = \nmpshape{M_2}{N_2}{P_2}\) -- so the best Winograd orbit member drops to a
single full product and a strictly smaller multiset whenever the split is unbalanced (the
gain vanishes only at equal parts). Concretely, \(\nmpshape{17}{17}{17}\) at the
\((9,8)^3\) allocation costs \(2930\) multiplications for
the best Winograd orbit member -- its lone all-small \(\nmpshape{8}{8}{8}\) product
replacing one \(\nmpshape{9}{9}{9}\) -- against \(2940\) for Strassen. It is a genuine
\emph{multiplication} saving, distinct from Winograd's lower addition count.

\subsection{Recombination with allocation}
\label{ssec:recombination}

The central composition operator in the catalog: the decomposition of
Section~\ref{ssec:decomposition} with freely chosen part sizes. A \emph{small} catalog
scheme --- the outer \emph{atom}, e.g.\ Strassen \(\nmpshape{2}{2}{2}\), typically \(2\) or
\(3\) per dimension --- multiplies a larger target \(\nmpshape{n}{m}{p}\): choose three
\emph{allocations} (a partition of \(n\) into \(n_o\) parts, of \(m\) into \(m_o\), of
\(p\) into \(p_o\)); each of the atom's \(r_o\) outer products becomes a smaller matmul on
the corresponding blocks, looked up in the catalog as a sub-scheme, and the assembled
algorithm has rank \(\sum_k r_{\text{sub-}k}\). Uniform allocation reduces to the Kronecker
product (Section~\ref{ssec:kronecker}); non-uniform allocation --- e.g.\
\(\nmpshape{17}{17}{17}\) via Strassen at \((9, 8)^3\) --- covers the gap left by Kronecker
chains, at the price of \emph{padding}: a sub-shape derived from two unequal blocks
inherits the larger block's size on the relevant axis.

\subsection{Projection margin: a death-rate score}
\label{ssec:projmargin}

Projection is the downward operator: restrict a larger scheme to a smaller shape and
dead-code-eliminate the products that die. A higher-rank parent can thereby project to a
\emph{lower}-rank child, with a clean closed form that turns ``keep the structurally better
scheme, not the lower-rank one'' into a computable criterion --- the downward analogue of
the bud score (Section~\ref{ssec:serendipitous}).

A product survives the restriction iff \emph{each} of its three factors \(U_k,V_k,W_k\)
still has a nonzero inside the kept sub-block; it dies as soon as \emph{any one} of them
vanishes there. Fix an axis, say the first (\(n\)); only \(U\) (input rows \(i\)) and
\(W\) (output rows \(i\)) carry an \(n\)-index. Write \(U_n(k)\) and \(W_n(k)\) for the
sets of \(n\)-indices column \(k\) of \(U\), resp.\ \(W\), touches. Dropping index \(i\)
kills product \(k\) exactly when \emph{either} its \(\alpha\)-form collapses
(\(U_n(k)=\{i\}\): all of \(U_k\) sits in the dropped input row) \emph{or} it writes only
the dropped output (\(W_n(k)=\{i\}\)). The \emph{private rank} at \(i\) is the count of
such products, \(\rho_n(i)=\#\bigl\{k:\ U_n(k)=\{i\}\ \text{or}\ W_n(k)=\{i\}\bigr\}\) (an
\emph{or}, not an \emph{and}: a product with \(U_n(k)=\{i\}\), \(W_n(k)=\{i'\}\) dies under
dropping \(i\) \emph{or} \(i'\)). Each axis uses two of the three matrices ---
\(n\!\to\!\{U,W\}\), \(m\!\to\!\{U,V\}\), \(p\!\to\!\{V,W\}\) --- and the \emph{projection
margin} of \(S\) (single index, best axis) is
\[
  \mu(S)=\max_{A\in\{n,m,p\}}\ \max_{i}\ \rho_A(i), \qquad
  R_{\text{after}} \;=\; R_{\text{before}}-\mu(S)
\]
exactly as a \emph{survivor count} (an upper bound on the child's rank, tightened below),
for projecting out the best single index; dropping \(\delta\) indices on one axis
generalises to \(\mu^{(\delta)}_A(S)=\max_{|D|=\delta}\#\{k:\mathrm{supp}_A(k)\subseteq
D\}\), with \(R_{\text{after}}=R_{\text{before}}-\mu^{(\delta)}\).

Projection therefore minimises \(R-\mu\), \emph{not} \(R\): a scheme is worth keeping over
a lower-rank rival whenever \(R_S-\mu(S)<R_{S'}-\mu(S')\), i.e.\ on the
\((\text{rank},\,\text{projection margin})\) Pareto frontier. The canonical instance is
\(\nmpshape{29}{30}{30}\): the rank-best parent, a flat \(\nmpcount{}{30}{30}{30}{12688}\)
(Schwartz--Zwecher 2025, ALS-style), has \(\mu=62\) and projects to \(12626\); the
higher-rank \emph{structured} \(\nmpcount{}{30}{30}{30}{12710}\) (Pan trilinear
aggregation) has \(\mu=122\) and projects to \(12588\), matching FMM. Structure localises
whole product-slabs to boundary slices (large \(\rho\)); a flat decomposition spreads
every product across all indices (small \(\rho\)), projecting worse despite lower rank.

\paragraph{Survivor count versus true projected rank.} The identity
\(R_{\text{after}}=R_{\text{before}}-\mu\) counts the products that \emph{survive} the
restriction; that count is exact, but only an \emph{upper bound} on the child's rank,
because restricting can make two survivors \emph{proportional} --- and a proportional pair
is one rank-one term, not two. The mechanism is a linear form losing a coordinate: suppose
two products agree on their \(\alpha\)- and \(\gamma\)-forms and differ only in their
\(\beta\)-forms, one reading \(b_{\cdot,k_1}+b_{\cdot,k_0}\) and the other
\(b_{\cdot,k_1}-b_{\cdot,k_0}\). Both survive the drop of output column \(k_0\) (each still
reads \(b_{\cdot,k_1}\)), yet \emph{after} the drop both compute \(b_{\cdot,k_1}\): the two
products coincide and fuse into one, for a \(-1\) that \(\mu\) never counts. (Over a
ternary base ``proportional'' forces ``\(\pm\)-equal'', so the fusion is a cheap sign
test.) This is the pair fusion of \S\ref{ssec:pair-fusion}, but \emph{induced} by the
projection rather than pre-existing: a drop can \emph{manufacture} the coincidence that
\S\ref{sec:nonoverlap} finds absent in un-projected schemes. The smallest instance is fully
explicit: our rank-\(11\) \(\nmpfield{\Rationals}{2}{2}{3}\) carries two products sharing the
input form \(A_{0,1}+A_{1,1}\) whose outputs differ only in column \(c=1\) --- one feeds
\(C_{0,0}\) and \(C_{0,1}\), the other only \(C_{0,0}\). Projecting to \(\nmpshape{2}{2}{2}\)
by dropping column \(c=1\) erases the distinguishing \(C_{0,1}\) entry; both products then
write \(C_{0,0}\) alone under the same input form and fuse, so nine survivors realise rank
eight. The effect stays small but is no toy-shape artefact ---
\(\nmpfield{\Rationals}{4}{7}{7}=144\) restricted to \(\nmpshape{4}{6}{7}\) leaves \(132\)
survivors but fuses to \(131\) --- and running the
ordinary fusion pass on the restricted scheme recovers it, turning the death-rate score
into the true projected rank. It also gives the rank-preserving meta-flip below a quantity
it can actually lower: support-unioning flips barely move survivor counts, but they do
create proportional pairs.

\paragraph{The opposite extreme: \(\mu=0\) and shared-rank neighbours.} Nothing forces a
drop to spare a multiplication. When no surviving product is \emph{private} to the dropped
slice --- every product touches a second kept index on that axis --- then \(\mu=0\) and the
projection inherits the parent's rank unchanged. The same cascade can do both: our
\(\nmpfield{\Rationals}{9}{17}{19}=1753\) drops one output column to
\(\nmpfield{\Rationals}{9}{17}{18}=1600\) (\(\mu=153\), a genuine reduction), yet the next
drop, to \(\nmpshape{9}{17}{17}\), kills nothing (\(\mu=0\)), so \(1600\) is the best known
for \emph{both} shapes. This is not the extra output column being free: it is only
\(R(\nmpshape{9}{17}{17})\le R(\nmpshape{9}{17}{18})=1600\) holding trivially (a projection
never raises rank) while no cheaper dedicated \(\nmpshape{9}{17}{17}\) scheme is known --- an
inherited \emph{upper bound}, not a proven rank, and the shared-rank plateau an artefact of
the catalog rather than of the shapes.

\paragraph{A small worked instance, and how to \emph{raise} the margin.} Ordering our
small catalog by \(\mu\), the high-margin schemes are uniformly the \emph{structured} ones
(Pan trilinear aggregation) and the flip-graph ones (Perminov), never the flat rank-minimal
decompositions: the Perminov \(\nmpcount{}{6}{6}{7}{183}\) has \(\mu=26\) on the \(p\)-axis
and projects to \(\nmpshape{6}{6}{6}\) at \(183-26=157\) (not the catalog optimum \(153\),
but the rank-minimal \(\nmpshape{6}{6}{7}\) scheme projects worse). At \emph{small}
dimension, however, the flatter low-rank cube still wins projection outright (the rank gap
dwarfs the margin gap); the structured-cube advantage only flips the winner once the ranks
converge, as at \(\nmpshape{30}{30}{30}\) above. Constructing a scheme of the same rank but
higher \(\mu\) is open (parallel to the bud question); a promising route is a
rank-preserving meta-flip, searched for the \(\mu\)-maximal representative of a rank class.
Here \(\mu\) \emph{selects} among existing schemes. Because projection strictly decreases
dimension, a new arrival with \(\mu>0\) propagates as a bounded, acyclic, event-driven
downward wave; how the search schedules these waves to a fixpoint is described in
Section~\ref{ssec:multipass}.

\subsection{Zero-peeling: input- and output-side (DIS09 \(\gamma_5\))}
\label{ssec:peel}

When a recombination over-allocates an axis and pads the excess with zeros, the padded
rows/columns are dead and the affected sub-products can be computed at a smaller shape.
Two directions, dual under the matmul tensor's \(S_3\) symmetry: \emph{input-side} (a
sub-product reads \(A\) only in some rows, or \(B\) only in some columns --- shrink its
input extents) and \emph{output-side} (its \(W\)-support lands only in positions peeled
away after un-padding --- shrink its output extents; the Drevet--Islam--Schost 2009 /
Islam 2009 \(\gamma_5\) reduction). The peel-aware recombination shrinks each sub-product
to the \emph{elementwise minimum} of both, and the block-split search enumerates the
(allocation, peel) patterns automatically (overshoot bounded per axis). The canonical
\(\gamma_5\) case -- \(\nmpshape{3}{3}{3}\) via padded \(\nmpshape{4}{4}{4}\) Strassen, the
corner product collapsing \(\nmpshape{2}{2}{2}{=}7 \to \nmpshape{1}{2}{1}{=}2\) -- is
closed and tested. Peeling is the second padding-fighter after axis-flip choice
(Section~\ref{ssec:axis-flip}); the two are complementary -- flip distributes padding
across products, peel removes it from a product whose slice was dead anyway -- and, like
axis-flip, it is \emph{not} a sharing-discovery mechanism (Section~\ref{sec:nonoverlap}).

\paragraph{Limit.} Peeling alone does not close the odd-cubic gaps such as
\(\nmpshape{17}{17}{17}{=}2934\) (FMM-Lille's 2026-06 holding; both catalogs have since
reached \(2930\)): there, no standard Strassen product has its support inside the peeled
corner, so peeling leaves us at \(2940\). FMM-Lille's \(2934\) recipe instead pair-fuses
the diagonal \(\nmpshape{9}{9}{9}\) products and substitutes the corner with
\(\nmpshape{8}{8}{8}\) -- a pair-fusion (Section~\ref{ssec:pair-fusion}) plus
outer-template substitution, not a peel.

\subsection{Pair fusion (Pan trilinear aggregation)}
\label{ssec:pair-fusion}

Pan 1980 \cite{pan1980} observed that two same-shape sub-products arising at distinct
positions in a larger scheme can be computed \emph{jointly} via a trilinear aggregation
taking \(a b c + a b + b c + c a\) bilinear products in place of the naïve \(2 a b c\);
cyclic pair-fusion identifies pairs of outer products with a compatible block-position
pattern and substitutes a fused sub-scheme. This is the first derivation strategy that
\emph{does} introduce sharing across outer products -- but only in a constrained,
pattern-matching way, not as a generic discovery mechanism.

\paragraph{When it pays: a rank-density / \(\omega\) threshold.}
Fusing is worthwhile only against the \emph{best} independent computation, i.e.
\(abc + ab + bc + ca < 2\,R(\nmpshape{a}{b}{c})\) -- not against the naïve \(2abc\).
Writing the leaf's \emph{rank density} \(\rho = R(\nmpshape{a}{b}{c})/(abc)\), this is
\(\rho > \tfrac{1}{2}\big(1 + \tfrac{ab+bc+ca}{abc}\big) \approx \tfrac{1}{2}\): pair-fusion
helps exactly when the leaf is still \emph{``fat''} -- its rank not yet driven below about
half the naïve \(abc\). For a cubic leaf \(\nmpshape{k}{k}{k}=R\) with implied exponent
\(\omega_{\mathrm{leaf}} = \log_k R\) (naïve \(R=k^3\Rightarrow\omega=3\)), the same
condition reads
\[
  \omega_{\mathrm{leaf}} \;>\; \omega^\star := \log_k\!\tfrac{k^3+3k^2}{2}
  \;=\; 3 - \log_k 2 + \log_k\!\big(1+\tfrac{3}{k}\big) \;\approx\; 3 - \log_k 2 .
\]
\emph{Worked examples.} \(\nmpshape{7}{7}{7}=249\) has \(\rho=249/343=0.726\) and
\(\omega_{\mathrm{leaf}}=\log_7 249=2.835\), above its threshold
\(\omega^\star=\log_7 245=2.827\): it \emph{fuses} (pair cost \(490 < 2\cdot 249 = 498\),
saving \(8\)). By contrast \(\nmpshape{8}{9}{9}=430\) has \(\rho=430/648=0.664\), just
below its threshold \(0.673\): \emph{not} fused (pair cost \(873 > 2\cdot 430 = 860\)).

\paragraph{Leaf-level pairing: any two slots of any recombination.}
Pair fusion originally applied only inside dedicated fused-recombination templates. The
device generalises much further, on a simple theoretical ground: the leaves of a
recombination are \emph{formally independent} bilinear problems, and bilinear schemes are
closed under input substitution --- so \emph{any} two same-shape cubic leaves of \emph{any}
recombination may be computed jointly at the pair cost \(k^3+3k^2\), whatever the combining
base. (Our engine initially gated pairing to naive-grid outer schemes, on the mistaken
belief that a non-trivial combining base could not carry a pair; the block-operand wiring
is in fact oblivious to how the slot's inputs are formed.) The search sweeps pairings over
every base of small rank, elects a pairing exactly when
\(k^3+3k^2 < 2\,R(\nmpshape{k}{k}{k})\) at the paired slots, and the lineage records the
base, allocations and slot matching (\code{R*[\,\ldots\,]} nodes). Two shapes this closed
at exact FMM-Lille ties: \(\nmpfield{\Rationals}{20}{23}{23}\), \(5945 \to 5906\) (two
\(\nmpshape{11}{11}{11}\) leaves paired at \(1694 < 2 \cdot 873\)), and
\(\nmpfield{\Rationals}{19}{19}{22}\), \(4538 \to 4536\).

This is why pair-fusion is a \emph{small-base, constant-factor} device: as \(k\to\infty\)
the threshold \(\omega^\star\to 3\) while good leaves have \(\omega_{\mathrm{leaf}}\) far
below \(3\), so the saving vanishes asymptotically; the catalog fuses only the two-product
instance of Pan's genuinely \(\omega\)-improving many-product aggregation.

\paragraph{Why pairs, not triples.} There is \emph{no} ``fuse three \emph{separate}
products'' primitive: the pair identity fuses two cyclically-related products
\(\nmpshape{a}{b}{c}+\nmpshape{b}{c}{a}\) at cost \(abc+ab+bc+ca\), while the
\(\tfrac{1}{3}\) of the full Pan/Islam aggregation comes from one product's own 3-fold
cyclic index symmetry, not from fusing three products. A single \(\nmpshape{n}{n}{n}\)
costs \(\approx n^3/3 + O(n^2)\), and the \(O(n^2)\) term is itself a \emph{best-of}
several Pan-family closed forms: \((n^3+12n^2+11n)/3\) (Islam 2009~\cite{islam2009}
Prop.~1, republished as DIS09~\S3) wins for moderate even \(n\); the smaller
\(\tfrac{15}{4}n^2\) quadratic of the \(\tfrac{45}{4}n^2\) branch (Hadas--Schwartz 1982,
sharpened by Schwartz--Zwecher 2025~\cite{schwartzzwecher2025}) overtakes it for even
\(n \gtrsim 28\); odd \(n\) uses \((n^3+15n^2+14n-6)/3\). The newest family member is LITA
(\emph{Local Improvements to Trilinear Aggregation}, Khoruzhii--Gel{\ss}--Pokutta
2026~\cite{khoruzhii2026lita}): closed forms \(R_e(N) = (4N^3+45N^2+116N+84)/12\) for even
\(N\) and \(R_o(N) = (4N^3+57N^2+14N-15)/12 - \lfloor 3(N-1)/8 \rfloor\) for odd \(N\),
valid for \(N \ge 19\), which in particular repair the historically weak \emph{odd} case
and reproduce several of FMM-Lille's large-cubic index values
(\(\nmpshape{21}{21}{21} = 5198\), \(\nmpshape{23}{23}{23} = 6586\), \ldots). The chooser
takes the \emph{minimum} over all of these branches, so each regime is accounted for, not
just the \(12n^2\) one. Figure~\ref{fig:ta-family-ranks} plots the whole family's
normalised cubic constant \(3R/N^3\) along \(N\): the crossovers show which aggregation is
tightest over which range --- LITA dominates for odd \(N \ge 19\) and large even \(N\),
while the Islam/DIS09 branch still wins at, e.g., \(N = 22\).

\begin{figure}[t]
  \centering
  \begin{tikzpicture}
    \begin{axis}[
        width=\linewidth, height=6.5cm,
        xlabel={cube side $N$}, ylabel={$3\,R(\langle N,N,N\rangle)/N^3$},
        xmin=1, xmax=64,
        legend pos=north east, legend cell align=left,
        grid=both, major grid style={black!10}, minor grid style={black!5},
        title={Trilinear-aggregation methods: normalised cubic constant (lower is better)}]
      \addplot[gray, mark=x, mark size=1.2pt, thick] coordinates {(2,3.7500) (4,2.9063) (6,2.5000) (8,2.2734) (10,2.1300) (12,2.0313) (14,1.9592) (16,1.9043) (18,1.8611) (20,1.8263) (22,1.7975) (24,1.7734) (26,1.7530) (28,1.7353) (30,1.7200) (32,1.7065) (34,1.6946) (36,1.6840) (38,1.6745) (40,1.6659) (42,1.6582) (44,1.6511) (46,1.6446) (48,1.6387) (50,1.6332) (52,1.6281) (54,1.6235) (56,1.6191) (58,1.6150) (60,1.6113) (62,1.6077) (64,1.6044)};
      \addlegendentry{TA\_pan (Pan 1980)}
      \addplot[blue, mark=*, mark size=1.2pt, thick] coordinates {(2,9.7500) (3,7.3333) (4,4.6875) (5,4.5120) (6,3.3056) (7,3.4111) (8,2.6719) (9,2.8313) (10,2.3100) (11,2.4748) (12,2.0764) (13,2.2340) (14,1.9133) (15,2.0604) (16,1.7930) (17,1.9296) (18,1.7006) (19,1.8274) (20,1.6275) (21,1.7454) (22,1.5682) (23,1.6781) (24,1.5191) (25,1.6220) (26,1.4778) (27,1.5745) (28,1.4426) (29,1.5336) (30,1.4122) (31,1.4982) (32,1.3857) (33,1.4672) (34,1.3625) (35,1.4399) (36,1.3418) (37,1.4155) (38,1.3234) (39,1.3937) (40,1.3069) (41,1.3741) (42,1.2920) (43,1.3563) (44,1.2784) (45,1.3402) (46,1.2661) (47,1.3254) (48,1.2548) (49,1.3119) (50,1.2444) (51,1.2995) (52,1.2348) (53,1.2880) (54,1.2260) (55,1.2773) (56,1.2178) (57,1.2674) (58,1.2102) (59,1.2582) (60,1.2031) (61,1.2496) (62,1.1964) (63,1.2416) (64,1.1902)};
      \addlegendentry{TA\_dis (Islam/DIS09)}
      \addplot[teal, mark=triangle*, mark size=1.2pt, thick] coordinates {(2,18.0000) (4,6.2344) (6,3.8889) (8,2.9590) (10,2.4720) (12,2.1753) (14,1.9767) (18,1.7284) (20,1.6459) (22,1.5800) (24,1.5263) (26,1.4816) (28,1.4438) (30,1.4116) (32,1.3836) (34,1.3593) (36,1.3378) (38,1.3187) (40,1.3017) (42,1.2864) (44,1.2725) (46,1.2600) (48,1.2485) (50,1.2380) (52,1.2284) (54,1.2195) (56,1.2113) (58,1.2036) (60,1.1965) (62,1.1899) (64,1.1837)};
      \addlegendentry{TA\_hs (Hadas--Schwartz 1982)}
      \addplot[orange, mark=square*, mark size=1.2pt, thick] coordinates {(2,17.2500) (4,6.0938) (6,3.8333) (8,2.9297) (10,2.4540) (12,2.1632) (14,1.9679) (18,1.7233) (20,1.6418) (22,1.5766) (24,1.5234) (26,1.4792) (28,1.4418) (30,1.4098) (32,1.3821) (34,1.3579) (36,1.3365) (38,1.3176) (40,1.3007) (42,1.2855) (44,1.2717) (46,1.2592) (48,1.2478) (50,1.2374) (52,1.2278) (54,1.2189) (56,1.2108) (58,1.2032) (60,1.1961) (62,1.1895) (64,1.1833)};
      \addlegendentry{TA\_sz (Schwartz--Zwecher 2025)}
      \addplot[red, mark=diamond*, mark size=1.2pt, thick] coordinates {(19,1.7565) (20,1.6376) (21,1.6838) (22,1.5733) (23,1.6239) (24,1.5206) (25,1.5736) (26,1.4768) (27,1.5310) (28,1.4397) (29,1.4942) (30,1.4080) (31,1.4621) (32,1.3805) (33,1.4339) (34,1.3565) (35,1.4091) (36,1.3353) (37,1.3868) (38,1.3165) (39,1.3669) (40,1.2997) (41,1.3489) (42,1.2846) (43,1.3327) (44,1.2709) (45,1.3178) (46,1.2585) (47,1.3042) (48,1.2472) (49,1.2918) (50,1.2368) (51,1.2803) (52,1.2272) (53,1.2697) (54,1.2184) (55,1.2599) (56,1.2103) (57,1.2507) (58,1.2027) (59,1.2422) (60,1.1957) (61,1.2342) (62,1.1891) (63,1.2268) (64,1.1829)};
      \addlegendentry{TA\_lita (KGP 2026)}
    \end{axis}
  \end{tikzpicture}
  \caption{Each trilinear-aggregation method gives an alternative rank for $\langle N,N,N\rangle$ over its own domain (TA\_pan/TA\_hs/TA\_sz are even only, TA\_hs/TA\_sz exclude the $N{=}16$ pole, TA\_lita requires $N\ge 19$). Plotted is the normalised leading constant $3R/N^3$; all converge to $1$, and the crossovers show which method is tightest over which range --- TA\_lita (KGP 2026) dominates for odd $N\ge 19$ and large even $N$, while TA\_dis still wins at e.g.\ $N{=}22$.}
  \label{fig:ta-family-ranks}
\end{figure}
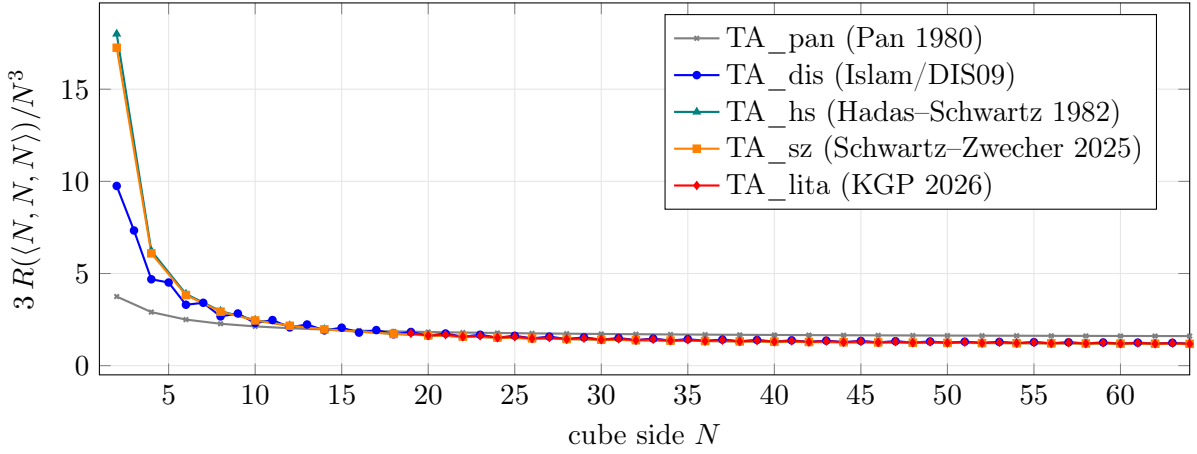
 So the real per-leaf choice is shallow pair fusion
(\(\approx k^3/2\) per product, light \(O(k^2)\) corrections) versus deep single-product TA
(\(\approx k^3/3\) per product, but \(\sim 4\times\) the \(O(k^2)\) corrections). At a
small leaf the corrections dominate and pairs win decisively. For the seven
\(\nmpshape{7}{7}{7}\) leaves of a \(\nmpshape{14}{14}{14}\) (or padded
\(\nmpshape{17}{17}{17}\)) recombination:

\begin{center}
\begin{tabular}{lrr}
\toprule
option & per product & total (7 leaves) \\
\midrule
solo (catalog SOTA) & \(249\) & \(1743\) \\
pair-fused & \(245\) & \(\mathbf{1719}\) \\
full single-product TA & \(390\) & \(2730\) \\
\bottomrule
\end{tabular}
\end{center}

\noindent Pair fusion (\(1719\)) beats both solo (\(1743\)) and full TA (\(2730\)). The
deep TA's \(\tfrac{1}{3}\) only overcomes its heavier corrections at \emph{even}
\(k \ge 18\); odd \(k\) historically carried the heavier \(15k^2\) term, until the LITA
closed forms~\cite{khoruzhii2026lita} repaired the odd case for \(k \ge 19\)
(Figure~\ref{fig:ta-family-ranks}). This is why the catalog pair-fuses small leaves and
reserves full TA for large ones; the chooser picks the cheapest of the three per leaf.

\subsection{Disjoint-sum decomposition ($\tau$-theorem)}
\label{ssec:disjoint-sum}

A target tensor can decompose as a sum of smaller matmul sub-tensors \emph{at the tensor
level} -- not at a single-axis level -- when it admits a $\tau$-theorem-style
decomposition: each leg computes a smaller matmul, and the legs combine via the additive
structure of the tensor. Schwartz--Zwecher 2025 \cite{schwartzzwecher2025} give explicit
constructive disjoint-sum decompositions for \emph{large} formats (e.g.\
\(n = 44, 56, 60\)).

We stress that \emph{we do not currently leverage the $\tau$-theorem anywhere}. The
disjoint-sum schemes we hold (Schwartz--Zwecher and those FMM/Perminov build this way) are
\emph{imported as explicit factor matrices}, not reconstructed from a $\tau$ decomposition:
no catalog entry carries a $\tau$ lineage, and every gap we initially attributed to a
$\tau$ mechanism turned out to close (or not) by other means -- Kronecker, recombination,
or a fortunate cousin scheme. A \code{DisjointSum} lineage node is \emph{reserved} in the
schema but exercised by no scheme today; we have no search operator constructing disjoint
sums, flagged in Section \ref{sec:openquestions} as the largest open extension.

\section{The non-overlap property and a methodological boundary}\label{sec:nonoverlap}

The composition operators of Sections \ref{ssec:axis-flip} through
\ref{ssec:peel} are subject to the following structural invariant.

\begin{theorem}[Non-overlap]\label{thm:nonoverlap}
Let \(O\) be an outer bilinear scheme of shape
\(\nmpshape{n_o}{m_o}{p_o}\) with \(r_o\) products, and let
\(L_1, \ldots, L_{r_o}\) be sub-schemes assigned by block-substitution
(the recombination-with-allocation operator,
Section~\ref{ssec:recombination}) to a target shape
\(\nmpshape{n}{m}{p}\) with allocations \((\mathbf{a}_A, \mathbf{a}_B,
\mathbf{a}_C)\) and an optional output peel pattern. The resulting
algorithm has rank exactly \(\sum_k r_{L_k}\), and no two of its
bilinear products are equal as bilinear forms.
\end{theorem}

\begin{proof}[Sketch]
The construction is block-substitution: each outer product
\(k \in [1, r_o]\) is instantiated as the sub-scheme \(L_k\) wired
into block-positions of \((A, B, C)\) determined by the allocations
and the outer scheme's bilinear form. Two columns of the resulting
\((U, V, W)\) can be equal as bilinear forms only if (a) they come
from the same outer product index AND (b) they correspond to the
same column of \(L_k\) -- which would require the same column of the
sub-scheme to be wired to two different positions, contradicting the
definition of block-substitution. The peel reduction
(Section \ref{ssec:peel}) removes dead rows of \(W\) but never
identifies two columns. The axis-flip orbit (Section
\ref{ssec:axis-flip}) is an outer transformation that does not change
the rank or the column-distinctness.
\end{proof}

\paragraph{Scope: the operations that \emph{reduce} below the sum.} The
theorem governs pure block-substitution: Kronecker, concatenation,
recombination-with-allocation, axis-flip, peel. Three operations
deliberately produce fewer than \(\sum_k r_{L_k}\) products, and none is a
counterexample, because each applies a \emph{known} reduction at predict
time; nothing is discovered during materialisation.
\begin{itemize}[itemsep=2pt]
  \item \textbf{Pair fusion (Pan TA)}: an optimisation of a fixed
    decomposition's product set through Pan's trilinear-aggregation
    identity (Section~\ref{ssec:pair-fusion}); the saving is a property of
    that identity, counted before materialisation.
  \item \textbf{The serendipitous product}: replaces an enlarged sub-block
    by its catalogued best rank --- a predict-time lookup.
  \item \textbf{Downward projection}: restricts a materialised parent and
    runs the same deterministic dead-code elimination that the
    projection-margin score (Section~\ref{sec:search}) already predicted.
\end{itemize}
So the deeper invariant below --- \emph{predicted rank equals materialised
rank, with no sharing found during materialisation} --- holds for these
three as well, even though the literal \(\sum_k r_{L_k}\) equality of the
theorem does not.

Theorem \ref{thm:nonoverlap} has two practical consequences worth
calling out.

\paragraph{Predicted rank equals materialised rank.} The exhaustive derivation
search (Section \ref{sec:search}) propagates rank predictions
\(\hat r = \sum_k r_{\text{sub-}k}\) computed by summing sub-shape
lookups in the catalog. The materialised scheme has exactly that
rank. There is no hidden discovery during materialisation; the
materialisation step is faithful execution of whatever strategy the
search picked. This is what lets us split the closure loop into a
fast rank-only search phase and a lazy materialise-only phase
without losing accuracy -- see Section \ref{sec:search} for the
algorithm.

\paragraph{A methodological boundary against hand-crafted schemes.}
Almost every hand-crafted small-format result in the literature
exists \emph{because} its author found a non-obvious shared
bilinear core: Strassen's \(\nmpfield{\Reals}{2}{2}{2} = 7\)
beats 8 because two of his seven products do double duty across
output cells; Laderman's \(\nmpfield{\Reals}{3}{3}{3} = 23\)
is the same story at larger scale; Smirnov, Pan, AlphaTensor,
AlphaEvolve all hunt for non-obvious cross-product sharing.
Theorem \ref{thm:nonoverlap} says our \emph{composition} operators do
\emph{not} hunt for such sharing -- the wins come from picking the
right composition of cataloged building blocks. (The reducing operators
of Section~\ref{ssec:decomposition} ff.\ --- serendipity, pair fusion,
projection --- do exploit known reductions, but never discover new ones;
see the scope remark above.)

This makes a ``composition matches hand-crafted rank'' outcome a
\emph{co-existence} claim, not a re-discovery claim. When the search
reached \(\nmpfield{\Rationals}{17}{17}{17} = 2930\) (2026-06) via a
Winograd recombination at \((9,8)^3\), FMM-Lille's then-best for the
same shape stood at \(2934\), by a route structurally distinct from
ours (on inspection, itself a plain recombination --- not the disjoint
sum its description suggested; see Section~\ref{ssec:disjoint-sum});
FMM has since matched \(2930\). Our scheme is a plain sum of
independent sub-products and reproduces none of any hand-crafted
construction's shared bilinear identities. The catalog records
distinct routes with distinct lineages and distinct
\code{attribution\_for\_rank}.

The two approaches are complementary, not redundant. Composition
expands the SOTA frontier without finding new bilinear identities;
hand-crafted constructions find new bilinear identities that can be
ingested as new outer schemes, immediately expanding what
composition can do.

\paragraph{The boundary, measured.} The 2026-07 cross-catalog
campaign (Section~\ref{sec:tables}) turned this methodological line
into an empirical one, in both directions.

\emph{Our schemes contain no accidental sharing.} The theorem leaves
one loophole open in principle: two \emph{independently built} leaves
of a recombination could, by coincidence, feed identical bilinear
forms over the target's variables --- a coinciding \(u \otimes v\) pair
would merge into one product (the \(W\)-columns add) for a free rank
\(-1\). We probed our largest contested scheme
(\(\nmpfield{\Rationals}{17}{17}{19} = 3267\), 3267 products) for such
pairs: \emph{zero} coinciding \(u \otimes v\) directions. The
invariant holds not only by construction within each substitution but
empirically across independent leaves.

\emph{What separates us from hand-crafted catalogs is exactly the
sharing the theorem forgoes.} At campaign end the catalog led or tied
FMM-Lille everywhere except six shapes (twelve ranks in total, out of
thousands of compared shapes). Every one of the six deficits traces to
the same device in the external artifact: a \emph{fusion kernel} ---
a set of products shared \emph{across} recursion slots. The simplest
worked specimen, from FMM's \(\nmpshape{8}{27}{30}\) artifact
(rank 3736, which we decoded as
\(\nmpshape{4}{9}{10}{:}250 \otimes \nmpshape{2}{3}{3}{:}15\) with a
modified kernel): the plain Kronecker product computes 250 independent
``slots'', each an instance of \(\nmpshape{2}{3}{3}\) costing 15
products. In the artifact, 229 slots are plain --- but five
\emph{pairs} of slots \((s,t)\) are computed jointly:
\[
  \underbrace{8}_{\text{plain products of } s}
  \;+\; \underbrace{8}_{\text{plain products of } t}
  \;+\; \underbrace{13}_{\text{mixed products serving both}}
  \;=\; 29
  \qquad\text{versus}\qquad
  15 + 15 = 30 \text{ independently.}
\]
Each mixed product's \(U\)-column combines contributions of both
slots (its column, reshaped as a base\(\times\)inner matrix, has rank
2 instead of the rank 1 of every plain Kronecker product --- this
``Kron-rank fingerprint'' is also how the device is detected, since
nothing at the lineage or rank-arithmetic level reveals it: the
sharing is visible only after expanding the scheme to explicit
\((U,V,W)\) and rank-testing each product). One rank saved per pair,
five pairs, plus related bookkeeping: rank 3736 where plain
composition gives 3750. This is Theorem~\ref{thm:nonoverlap}'s
boundary made concrete --- the external catalog's residual edge over
our composition closure is \emph{only} this cross-slot sharing device,
and conversely no engine of Section~\ref{sec:strategies} can reproduce
it without a new, dedicated constructor
(Section~\ref{sec:openquestions}).

\paragraph{The invariant is current-state, not eternal.}
Theorem \ref{thm:nonoverlap} depends on the materialiser being pure
block-substitution: post-construction column fusion, random sign-salting
of sub-blocks, or coefficient mixing between outer products sharing a
sub-shape would each let materialised rank drop \emph{below} predicted
rank --- a feature, but one that breaks the invariant the two-phase
search of Section \ref{sec:search} relies on, so any such ``salt'' must
be confined to the final materialisation step or come with re-derived
predictions. We decided against salt for now: a pure search-time
prediction is easier to reason about and to cache, and avoids subtle
correctness bugs at the search/materialise boundary; sharing-style
discoveries are deferred to pair fusion (Section \ref{ssec:pair-fusion}),
and the disjoint-sum / $\tau$ route (Section \ref{ssec:disjoint-sum})
remains future work.

\section{Recombination multisets of a base}\label{sec:multisets}

Recombining a base at a block decomposition
(Section~\ref{ssec:recombination}) sends each of its \(r\) products to a
smaller multiplication; the resulting \emph{multiset} of \(r\) sub-shapes
is, for plain additive recombination, a complete rank invariant: the
recombined rank is \(\sum_k r_{\text{sub-}k}\). The multiset depends only
on the support pattern of \((U,V,W)\), not on coefficients, and --- with
each axis's blocks ordered descending --- is \emph{symbolic} in the block
sizes \((n_i,m_j,p_k)\); the split \((9,8)^3\to\nmpshape{17}{17}{17}\) used
below is one instantiation. This section enumerates, by derivation rather
than search, which multisets a fixed base can ever produce.

\paragraph{Relation to Sedoglavic's Proposition~1.}
The enumeration recovers, and strictly generalises, the closed form of
Sedoglavic~\cite{sedoglavic2017}, whose Proposition~1,
\[
  \nmpshape{u{+}v}{u{+}v}{u{+}v}\ \le\
  \nmpshape{u}{u}{u}\;+\;3\,\nmpshape{u}{u}{v}\;+\;3\,\nmpshape{v}{v}{u}
  \qquad(u>v),
\]
is exactly the recombination multiset of the Strassen \nmpshape{2}{2}{2}
base at the two-part allocation \([u,v]\) on every axis (one cube plus two
cyclic triples); the headline cases
\((u,v)=(4,3)\to\nmpshape{7}{7}{7}{=}250\) and
\((9,8)\to\nmpshape{17}{17}{17}{=}2940\) are two of its allocations. The
enumeration subsumes the formula in three directions: every base, every
allocation (unbalanced and \(\ge3\)-part splits included), and --- via the
dominance frontier of Section~\ref{ssec:frontier} --- the optimal multiset
rather than one symbolic recipe. The excluded case \(u=v\) is recovered as
Pan trilinear aggregation~\cite{pan1978}, kept only when cheaper; and the
generalisation can beat the proposition at its own shapes: for
\nmpshape{17}{17}{17} an allocation in the Strassen orbit dominates the
\([9,8]\) split, \(2930<2940\). The closed form is thus not a separate
construction we wire in, but a single base-and-allocation point of the
multiset frontier search.

\subsection{The orbit and a decoupling lemma}

By de Groote~\cite{degroote1978} the rank-7 decompositions of \nmpshape{2}{2}{2}
form a \emph{single} \(\mathrm{GL}_2(\Rationals)^3\) orbit (its discrete
isotropy is analysed in~\cite{burichenko2014}). Every scheme is therefore a
change of basis of Strassen's, acting on the factors by
\(U'_k = X^{\!\top}U_kY^{\!\top}\),
\(V'_k = Y^{-\top}V_kZ^{\!\top}\),
\(W'_k = X^{-1}W_kZ^{-1}\)
for \((X,Y,Z)\in\mathrm{GL}_2\times\mathrm{GL}_2\times\mathrm{GL}_2\). A direct
brute-force sweep of ternary schemes is hopeless (\(\approx 39\text{k}\)
candidate rank-one terms per product, \(39\text{k}^7\) paths). The structure that
rescues it:

\begin{lemma}[Per-axis decoupling]\label{lem:decouple}
The recombination sub-dimension on a given axis depends on \emph{only one} of
\(X,Y,Z\), and only through that matrix's column/row directions. Concretely, on
the first axis,
\[
  \mathrm{sub}A_k = \mathrm{big} \iff
  d_0^{\!\top}U_k\neq 0 \ \wedge\ \mathrm{perp}(d_1)^{\!\top}W_k\neq 0,
  \qquad X=[\,d_0\mid d_1\,],
\]
and analogously for the second and third axes via \(Y\) and \(Z\) --- with the
caveat that, by the transpose placement of the action, \(Y\) and \(Z\) act
through their \emph{row} directions (and the dual tests use adjugate
\emph{columns}), so the displayed column formula must not be transplanted
verbatim.
\end{lemma}

The lemma holds because an invertible factor on the \emph{other} two axes kills
no rows or columns of a factor matrix, so the per-axis support is unaffected by
it. Consequently the three axes are independent and the set of realisable
multisets is the product of three per-axis pattern sets, zipped by product index.

\subsection{Exact enumeration}

Each per-axis pattern is piecewise-constant in the direction(s) and changes only
when a direction meets a null-space of one of the fixed (integer) base factors.
The critical directions are therefore finite and rational, so a finite
integer-direction sweep enumerates \emph{every} realisable pattern --- exactly,
with integer arithmetic (adjugates in place of inverses, so zero/non-zero tests
are exact). For a two-part axis this is provably complete: the factor
null-spaces are the only critical directions, and one generic direction covers
the open stratum. For \(\ge 3\)-part axes completeness is certified only
\emph{empirically}, by direction-bound stability --- a bound, not a proof. We
emphasise that this is a \emph{derivation}, not a search over schemes.

\paragraph{Result for \nmpshape{2}{2}{2}.} Over the full
\(\mathrm{GL}_2(\Rationals)^3\) orbit there are exactly \textbf{40} distinct
canonical recombination multisets (canonical = quotiented by the base's
\(S_3\) axis-symmetry). The same set covers every rank-7 scheme over any
characteristic-0 field: the per-axis strata are cut by \emph{rational}
hyperplanes (the null-spaces of the integer Strassen factors), so every
stratum nonempty over \(\Reals\) or \(\Complex\) already contains a rational
point. The count is stable across direction bounds and identical whether
seeded from Strassen or Winograd --- an independent check of de~Groote's
single-orbit theorem. Of the 40, \textbf{6} are realised by ternary
\(\{-1,0,1\}\) schemes: those of Strassen, Winograd,
AlphaTensor-\(\Integers\), and Perminov, plus two we register as bases so
the search can use them --- one matching AlphaTensor-\(\Ftwo\)'s support,
and a novel base whose multiset is
\(3\cdot\nmpshape{n_1}{m_1}{p_1}+\nmpshape{n_1}{m_2}{p_2}+\nmpshape{n_2}{m_1}{p_2}+\nmpshape{n_2}{m_2}{p_1}+\nmpshape{n_2}{m_2}{p_2}\),
of independent interest because its triple cubic group admits 3-way
trilinear aggregation~\cite{pan1978}, a cost lever Strassen's own multiset
does not expose. The ternary count is a lower bound (``at least 6''):
stable under widening the change-of-basis alphabet and the choice of seed,
but not certified exhaustive over all ternary schemes. It coincides with
the size of the canonical frontier below --- suggestive, not proven
set-identical.

\subsection{Dominance pruning: the frontier}\label{ssec:frontier}

Most canonical multisets can never be rank-optimal and need never be considered.
The recombined rank is \(\sum_k r_{\text{sub-}k}\), and matrix-multiplication rank
is \emph{monotone} --- \(\operatorname{rank}\nmpshape{a}{b}{c}\le\operatorname{rank}\nmpshape{a'}{b'}{c'}\)
whenever \((a,b,c)\le(a',b',c')\) componentwise (pad the smaller product into the
larger). The cost is thus a sum over products of a monotone term, so sub-shape
dominance is cost-monotone: if every product of multiset \(A\) can be matched to a
distinct product of \(B\) with \(A\)'s sub-shape componentwise \(\le\) its partner,
then \(\mathrm{cost}(A)\le\mathrm{cost}(B)\) for \emph{every} base recursion and at
\emph{every} allocation (block index \(0\) = largest block under any split). Only the
\emph{dominance frontier} --- the non-dominated antichain --- can be optimal; pruning
to it is lossless for the minimum and allocation/base-independent. This is the exact
formalisation of the folklore ``\(1\cdot\mathrm{BBB}+2\cdot\mathrm{SSS}\) beats
\(2\cdot\mathrm{BBB}+1\cdot\mathrm{SSS}\)'': the second is dominated.

Two quotient levels matter, and they differ. The \emph{canonical} frontier
quotients by the base's shape automorphisms (dominance is then
axis-permutation aware, since rank is); it is exact when the allocation is symmetric
or the search also permutes axis assignments. The finer \emph{axis-tagged} frontier
applies no quotient and is lossless even for asymmetric allocations, where axis
identity is real. Both are computed from the canonical \emph{keys}, not from any
seeded representative, so both are seed-independent.

\begin{center}
\begin{tabular}{lrrrr}
\toprule
 & \multicolumn{2}{c}{canonical (mod axis sym.)} & \multicolumn{2}{c}{axis-tagged} \\
\cmidrule(lr){2-3}\cmidrule(lr){4-5}
base & multisets & frontier & multisets & frontier \\
\midrule
\nmpshape{2}{2}{2} & 40     & \textbf{6}   & 144     & \textbf{12} \\
\nmpshape{2}{2}{3} & 12\,217 & \textbf{21}  & 23\,719 & \textbf{33} \\
\nmpshape{2}{3}{3} & 62\,487 & \textbf{170} & 124\,250 & \textbf{310} \\
\bottomrule
\end{tabular}
\end{center}

Two observations. First, the frontier is far more robust than the raw
enumeration: for \nmpshape{2}{2}{3} the canonical count still grows with the
direction bound (\(8\,425\to 12\,217\)) while the frontier holds at \(21\) across
bounds --- so the pruned object is stable even where completeness of the full set is
only an empirical bound. Second, the frontier cannot collapse to a single point: a
product is one corner-type across all three axes at once, so the axes are coupled and
several incomparable trade-offs survive (this is exactly why
\nmpshape{17}{17}{17} admits the spread of ranks \(2930,\dots,2958\) from one base).

Operationally this makes selection cheap. For a fixed allocation the cost is
\emph{linear in the corner-type count-vector},
\(\mathrm{cost}=\sum_{\text{type}}\mathrm{count}[\text{type}]\cdot
\operatorname{rank}(\text{type})\); the frontier is allocation-independent,
so it is built once and each allocation is scored by one
count-vector\,\(\times\)\,rank product over the \(\le\) few-hundred frontier
members. The frontier thus collapses the multiset axis of the search to a
small fixed table, leaving the allocation sweep as the only combinatorially
large dimension.

\subsection{Generality}

The construction is base-agnostic: the same orbit enumeration
(\code{RecombinationMultisetOrbit}) applies to any base \nmpshape{n}{m}{p}
over \(\mathrm{GL}_n\times\mathrm{GL}_m\times\mathrm{GL}_p\), canonicalised
by the base's shape automorphisms (the full \(S_3\) for a cube, one swap
for \nmpshape{2}{3}{3}, trivial for an all-distinct shape). It thus
answers, for any candidate outer base, \emph{which} sub-shape multisets ---
hence which trilinear-aggregation and padding-avoidance opportunities ---
that base can ever expose to the frontier search.

\section{A constructive realisation of the Hopcroft--Kerr bound}\label{sec:hk71}

Hopcroft and Kerr~\cite{hopcroftkerr1971} proved
\(R(\nmpshape{2}{p}{n}) \le \lceil (3pn + \max(p,n))/2 \rceil\)
over any field, by an explicit construction (their Theorem~1). The published
proof, however, compresses its hardest steps. Explicit in the 1971 paper:
the three diagonal methods; Lemma~2's three \emph{different}-method pair
cases; \emph{one} same-method pair case, \((1,1,\mathrm{bridge}\text{-}2)\);
the non-constructive existence proofs of Lemmas~1 and~3; and the Case-2
outline with the \(Z\)-aggregation identities. Not in the paper: the
remaining \emph{five} same-method cases (``the other cases follow by
symmetry'' --- they do not: the natural involutions swap methods 1 and 2
but fix method~3); any concrete augmentation matrix; and the fact that
Lemma~3's own sequence forces \((3,3)\)-same-method pairs, a case for which
no derivation is given and --- as we show below --- none exists within the
natural product family.

For five decades the bound has effectively been a black box. The community
catalogs (FMM-Lille~\cite{fmmlille}, Perminov~\cite{perminov}) credit
Hopcroft--Kerr for many \(\nmpshape{2}{p}{n}\) entries, yet at several
shapes — as of our 2026-06 full synchronisation of both catalogs —
\emph{neither holds a scheme attaining the formula, nor any reference to
its value as an attainable bound}: at \(\nmpshape{2}{10}{15}\) both publish
234 against the formula's 233; at \(\nmpshape{2}{12}{16}\) the best
published is 298 against 296; at \(\nmpshape{2}{10}{16}\) (formula 248) no
attaining scheme is published at all. If a working constructive reading of
the 1971 paper existed anywhere, these entries would attain the formula.
This section gives a complete, machine-verified constructive realisation,
identifies precisely which ingredient resisted (and why), and produces
schemes strictly below every published catalog at exactly those shapes.

\subsection{The construction, in blocks}

Write \(Y = \bar A X\) with \(A \in K^{p\times 2}\) augmented to
\(\bar A = MA \in K^{n\times 2}\), \(X \in K^{2\times n}\), and
\(p \le n \le 2p-1\). The constructive pipeline decomposes into seven blocks.
Blocks (i)--(iii) are Hopcroft--Kerr's (with (iii) only partially explicit in
the paper); blocks (iv)--(vii) are this work's, and are what turn the 1971
existence argument into an algorithm.

\paragraph{(i) Diagonal methods.} Each internal row \(i\) carries a method
\(c(i)\in\{1,2,3\}\) computing \(y_{ii}\) with two reusable products:
method 1 uses \(\{A{=}a_2(x_1{+}x_2),\,B{=}(a_1{-}a_2)x_1\}\), method 2
\(\{C{=}(a_1{-}a_2)x_2,\,D{=}a_1(x_1{+}x_2)\}\), method 3
\(\{E{=}a_2x_2,\,F{=}a_1x_1\}\).

\paragraph{(ii) Different-method pairs (HK Lemma 2).} A pair
\((y_{ij}, y_{ji})\) with \(c(i)\ne c(j)\) costs three new products. A
structural fact we use heavily: the \((\alpha,\beta)\)-method pair emits the
\emph{third} method's two products on a signed combination of the rows
(\((1,2)\to E,F\) on \(a_i{+}a_j\); \((1,3)\to C,D\) and \((2,3)\to A,B\) on
\(a_j{-}a_i\)) — each pair manufactures a \emph{virtual diagonal} for free.

\paragraph{(iii) Same-method pairs via a bridge.} When \(c(i)=c(j)\), three
new products suffice by routing through a third row \(b\), consuming the
diagonal and pair products of \((i,b)\) and \((b,j)\). Which bridge methods
are derivable is the crux below.

\paragraph{(iv) Bridge selection.} In cyclic distance-ordered processing,
\emph{any} arc-interior position \(b = i+e \ (1 \le e < d)\) is a legal
bridge: both required pairs are already computed. With an alternating
\(1,2\)-coloring an opposite-method interior position essentially always
exists, so the benign cases \((1,1,\mathrm{b}2)\), \((2,2,\mathrm{b}1)\)
carry almost all pairs.

\paragraph{(v) Unimodular Lemma-1 matrices: integer schemes.}
\(M\) needs every cyclic \(p\)-window nonsingular; we strengthen this to
\emph{unimodular} (every window determinant \(\pm 1\)), which makes the
back-substituted \(W\) — hence the whole scheme — \emph{integer}, over
\(\Integers\) rather than \(\Rationals\) (the classical Vandermonde rows
are numerically unusable here, and dense \(\pm1\) rows can \emph{never} be
unimodular for \(n-p \ge 2\), every \(2\times2\) \(\pm1\)-minor being
even). Window determinants of \([I_p; B]\) reduce to three minor families
of the augmented block \(B\) (\(m = n-p\) rows): leading, trailing, and
sliding \(m\)-column minors. A \emph{Euclidean construction} satisfies all
three with pure \(0/1\) entries and no search: a period-\(m\) comb body
whose sliding windows are permutation matrices, plus — for
\(r = p \bmod m > 0\) — a tail of \(r\) columns equal to
\(B(m, r)^{\!\top}\). Eliminating the comb's permutation part shows the
crossing and trailing minors of \(B\) are exactly the same three families
of the tail transposed, so the recursion descends like the gcd and
terminates at the pure comb. Every window is still certified by
fraction-free BigInteger elimination, and the back-substitution
\(A x_j = M_{\mathrm{win}(j)}^{-1}(\bar A X)_{\mathrm{win}(j),j}\) is exact
integer arithmetic.

\paragraph{(vi) Case 2: circulant matching and the repaired Step 3.} For even
\(p = 2k+2\), the band of half-width \(k\) leaves each column one cell short;
the budget — \(\lceil(3pn{+}n)/2\rceil - n(3k{+}2)\) — is exactly \(n/2\)
full pairs at distance \(k{+}1\): a perfect matching in the circulant
\(i \leftrightarrow i{+}k{+}1 \pmod n\), whose \(\gcd(n,k{+}1)\) orbit
cycles are matched edge-alternately. Odd cycles each leave one column,
completed in pairs by Hopcroft--Kerr's \emph{Z-trick} (their pp.~12--13),
which we restate operationally: for leftover columns \(i_2, i_4\) at
separation \(\delta \le k\), set \(i_1 = i_2{+}k{+}1\), \(i_3 = i_4{-}k{-}1\)
and treat \(\alpha = a_{i_2}{+}a_{i_3}\), \(\beta = a_{i_1}{-}a_{i_4}\) as
virtual rows: their method products already exist as the band pairs'
cross-products (by the structural fact in (ii)), so the virtual Lemma-2 cross
pair \(Z(\beta,\alpha_x), Z(\alpha,\beta_x)\) costs three new products and
yields both missing cells after subtracting cells of \emph{complete} columns
— including out-of-band cells, recovered exactly as rational combinations
through the Lemma-1 window relation. One leftover (n odd) is absorbed by the
formula's ceiling.

\paragraph{(vii) Chained augmentation for \(n > 2p{-}1\).} One Lemma-1 band
covers at most \(2p{-}1\) columns, so for larger \(n\) the columns are
partitioned into segments of size \(s \in [p, 2p{-}1]\), each built by the
band construction and concatenated along the column axis (the products of
distinct segments share nothing; ranks add). Because \(\max(p,s) = s\)
throughout the segment range, per-segment formula values telescope to the
global \(\lceil n(3p{+}1)/2 \rceil\) whenever the ceiling slack vanishes:
automatic for odd \(p\) (each segment cost \(s(3p{+}1)/2\) is an integer),
and for even \(p\) achievable by using at most one odd-size segment, of
parity matching \(n\). Rather than hand-coding these rules, the partition is
chosen by a dynamic program over the \emph{achieved} ranks of the segment
builds, which also routes around the degraded \(g \ge 6\) segment sizes of
block (vi) — e.g.\ \(p{=}12, n{=}36\) splits as \(16{+}20\) (both exact), not
\(18{+}18\) (both \(+1\)). The chain attains the formula at \emph{every}
swept \(n > 2p{-}1\), all parities.

\subsection{The bridge-3 cases: a published impossibility, and its
operational reversal}

The repository's prior analysis proved (Gr\"obner-verified): \emph{no three
rank-1 products complete the \((2,2,\mathrm{bridge}\text{-}3)\) case over
characteristic 0 when the reusable set is}
\(S = \{C_i, D_i, C_j, D_j, E_b, F_b\}\). That theorem is correct — and
turns out not to bind. The emission actually disposes of \emph{twelve}
reusables: the three diagonals \emph{and all products of the Lemma-2 pairs
\((i,b)\) and \((b,j)\)}, each independently weightable in the output
combination. Over this true set, exact search (120 candidate atoms, rational
arithmetic) finds three-product completions; solving for the weights gives
the explicit identities
\begin{align*}
(2,2,\mathrm{b}3):\quad
y_{ij} &= -C_i + D_i + E_b - A(\delta_{ib}) + G_{ib} + A(\delta_{bj})
          - E(\sigma) - G^\ast,\\
y_{ji} &= -C_i + D_i + F_b - B(\delta_{ib}) - G_{ib} + B(\delta_{bj})
          - F(\sigma) + G^\ast,
\end{align*}
with \(\delta_{ib} = a_b - a_i\), \(\delta_{bj} = a_b - a_j\),
\(\sigma = a_i + a_b - a_j\) (mirrored on the \(x\)-side), and the single
genuinely new mixed product
\(G^\ast = (a_{i,1} + a_{b,2} - a_{j,2})\,(x_{1,b} - x_{1,j} - x_{2,i})\);
the three \emph{new} multiplications are \(E(\sigma), F(\sigma), G^\ast\).
An analogous identity closes \((1,1,\mathrm{b}3)\). This vindicates
Hopcroft--Kerr's ``three additional multiplications'' claim for the bridge-3
cases: the published impossibility holds for its narrow \(S\), but the
construction was never confined to it.

\subsection{The genuine gap, and a geometric limit}

The same exact search over the true reusable set finds \emph{no}
three-product completion for \((3,3,\mathrm{bridge}\text{-}1)\) or
\((3,3,\mathrm{bridge}\text{-}2)\) — and here the negative can be promoted
from a catalog falsification to a structural theorem. Block-decompose the
monomial space by rows \(\{a_i,a_b\}\) versus \(a_j\) and columns
\(\{x_i,x_b\}\) versus \(x_j\): the two bridge-sum targets sit, with rank 2,
in blocks no reusable product touches, which forces \emph{any} three rank-1
completing atoms (arbitrary linear forms on rows and columns \(\{i,b,j\}\) —
no candidate catalog) to be supported on all four blocks, pins their
coefficient supports, and collapses their spill into the
\((\{a_i,a_b\},\{x_i,x_b\})\) block to two nonzero dyads with factors in the
virtual-row space \(\langle a_i{+}a_b\rangle\) and virtual-column space
\(\langle x_i{+}x_b\rangle\), each required to lie in the reusable span. An
exact computation shows that intersection is one-dimensional, generated by
the \emph{rank-2} virtual diagonal
\((a_{i1}{+}a_{b1})\otimes(x_{1i}{+}x_{1b}) +
(a_{i2}{+}a_{b2})\otimes(x_{2i}{+}x_{2b})\) — it contains no rank-1 element,
so no three-product completion exists, for either bridge method
(scope: the emitter's nine-product reusable set and atoms local to
\(\{i,b,j\}\)). This is the robust obstruction:
Hopcroft--Kerr's own Lemma-3 method sequence places method-3 rows at
spacing~2, forcing \((3,3)\) same-method band pairs — so \emph{the paper's
Theorem-1 proof, as written, rests on an underivable case}. Our construction
avoids it by keeping method-3 rows pairwise more than \(k\) apart, which is
possible except at a sharply characterised family: for \(n = 3(k{+}1)\) the
circulant decomposes into \(k{+}1\) triangles, demanding
\(\lfloor (k{+}1)/2 \rfloor\) Z-pairs; three method-3 rows pairwise
\({>}k\) apart on \(C_{3(k+1)}\) would need their arcs to sum to at least
\(3k{+}6 > n\) — impossible. Those shapes (\(g \ge 6\) odd cycles; six in
the full sweep range, \(\nmpshape{2}{12}{18}\) through
\(\nmpshape{2}{24}{30}\)) land at formula\(+1\) to \(+3\); everything else
closes exactly.

\subsection{A worked obstruction: the gcd accounting at
\(\nmpshape{2}{12}{18}\)}
\label{ssec:hk71-gcd}

The limit just stated is, at bottom, a counting argument, and it is worth
making explicit: it is the only place the construction misses, and it
localises the miss to a single arithmetic clash between two demands on the
same parameter \(k{+}1\).

\paragraph{The ledger has no slack.} For even \(p = 2k+2\) the band of
half-width \(k\) costs \(2n\) diagonal products plus \(3kn\) for the pairs at
distances \(1..k\), i.e.\ \(n(3k{+}2)\), and supplies \(2k{+}1 = p{-}1\) band
cells per column. The formula then allows exactly
\(\lceil (3pn{+}n)/2 \rceil - n(3k{+}2) = 3n/2\) further products for the
\(n\) remaining cells --- precisely \(n/2\) pairs at three products each,
with no residue. Everything below is therefore a statement about whether a
certain pairing exists; a column whose last cell is instead computed
naively costs two products against its \(3/2\) share, so \emph{each pair of
columns that fails to be paired costs exactly \(+1\)}.

\paragraph{Demand: a gcd.} The distance-\((k{+}1)\) pair \((i, i{+}k{+}1)\)
supplies the missing cell of \emph{both} its columns, so the required pairing
is a perfect matching in the circulant \(i \leftrightarrow i{+}k{+}1
\pmod n\). That graph is a disjoint union of
\[
  g \;=\; \gcd(n,\, k{+}1) \quad\text{cycles of length } n/g .
\]
Even cycles match perfectly. If \(n/g\) is odd, every cycle strands one
vertex, leaving \(g\) \emph{leftover columns}; these are completed two at a
time by the \(Z\)-trick of block (vi), so the construction needs
\(\lfloor g/2 \rfloor\) \(Z\)-pairs (a single odd leftover forces one naive
cell, but \(n/g\) and \(g\) both odd makes \(n\) odd, and the formula's
ceiling absorbs it).

\paragraph{Supply: an arc-sum.} Each \(Z\)-pair introduces one method-3 row.
By the \((3,3,\mathrm{bridge})\) theorem above, two method-3 rows may not sit
within band distance of each other --- their arc must be at least
\(k{+}2\) --- or the band contains a \((3,3)\) same-method pair with no
three-product completion. On the cycle \(C_n\) the arcs sum to \(n\), so at
most
\[
  \left\lfloor \frac{n}{k+2} \right\rfloor \quad\text{method-3 rows fit.}
\]

\paragraph{The clash.} Composing the two counts, the construction lands at
\[
  \text{formula} \;+\; \max\!\left(0,\;
    \left\lfloor \frac{g}{2} \right\rfloor -
    \left\lfloor \frac{n}{k+2} \right\rfloor \right),
  \qquad g = \gcd(n,\,k{+}1),
\]
whenever \(n/g\) is odd (and at the formula otherwise). This is not an
independent result --- it is the composition of the \(+1\)-per-dropped-pair
ledger with the arc-sum bound --- but it reproduces every observed excess in
the sweep, and it exposes why the failure is structural rather than
incidental: \(k{+}1\) simultaneously \emph{creates} the demand (through
\(g = \gcd(n, k{+}1)\)) and, via the band width, \emph{throttles} the supply
(through the spacing \(k{+}2\)). The two move together, in opposite
directions. The worst case is \(n = 3(k{+}1)\), where the circulant is all
triangles: demand \(\lfloor (k{+}1)/2 \rfloor\) grows with \(k\) while supply
stays at \(\lfloor 3(k{+}1)/(k{+}2) \rfloor = 2\) for every \(k \ge 2\).

\paragraph{\(\nmpshape{2}{12}{18}\), end to end.} Here \(k = 5\), \(n = 18\),
formula \(\lceil (3 \cdot 12 \cdot 18 + 18)/2 \rceil = 333\).

\begin{center}
\begin{tabular}{lrr}
\hline
step & products & cumulative \\
\hline
diagonals + pairs at distance \(1..5\) & \(36 + 270\) & 306 \\
\(\;\;\)(leaves 11 of the 12 band cells per column) & & \\
\hline
budget remaining to the formula & \(27 = 3n/2\) & 333 \\
\hline
matching: \(g = \gcd(18,6) = 6\) triangles, one leftover each & & \\
\(\;\;\)6 matched pairs (12 columns) & \(18\) & 324 \\
\(\;\;\)6 leftover columns, demand \(\lfloor 6/2 \rfloor = 3\) \(Z\)-pairs
  & \((9)\) & \((333)\) \\
\(\;\;\)supply \(\lfloor 18/7 \rfloor = 2\) \(Z\)-pairs (4 columns)
  & \(6\) & 330 \\
\(\;\;\)2 stranded columns, naive & \(4\) & \textbf{334} \\
\hline
\end{tabular}
\end{center}

The arithmetic closes at 333 --- three \(Z\)-pairs cost exactly the 27
available --- and fails on geometry: three method-3 rows pairwise at least
\(7\) apart on \(C_{18}\) need arcs summing to \(21 > 18\). Two fit (rows
\(0\) and \(9\)); the third \(Z\)-pair is dropped, its two columns fall back
to two products each, and the scheme lands at 334. The same ledger runs for
the whole family:

\begin{center}
\begin{tabular}{lrrrrrr}
\hline
shape & \(k{+}1\) & \(g\) & demand & supply & excess & emitted \\
\hline
\(\nmpshape{2}{10}{15}\) & 5 & 5 & 2 & 2 & 0 & 233 \\
\(\nmpshape{2}{12}{18}\) & 6 & 6 & 3 & 2 & \(+1\) & 334 \\
\(\nmpshape{2}{14}{21}\) & 7 & 7 & 3 & 2 & \(+1\) & 453 \\
\(\nmpshape{2}{24}{30}\) & 12 & 6 & 3 & 2 & \(+1\) & 1096 \\
\(\nmpshape{2}{16}{24}\) & 8 & 8 & 4 & 2 & \(+2\) & 590 \\
\(\nmpshape{2}{18}{27}\) & 9 & 9 & 4 & 2 & \(+2\) & 745 \\
\(\nmpshape{2}{20}{30}\) & 10 & 10 & 5 & 2 & \(+3\) & 918 \\
\hline
\end{tabular}
\end{center}

\(\nmpshape{2}{10}{15}\) is included as the boundary case that \emph{does}
close: \(g = 5\) demands two \(Z\)-pairs, two fit, and the odd leftover is
absorbed by the ceiling --- which is exactly why the flagship
\(\nmpfield{\Integers}{2}{10}{15}{=}233\) exists while
\(\nmpshape{2}{12}{18}{=}333\) does not. The escape routes are those listed
in the honesty notes below (a twelve-product \((3,3,\mathrm{bridge})\)
completion, non-local atoms, or a four-product completion traded against a
saving elsewhere); note also that no published catalog holds a scheme at
these formula values either --- FMM-Lille's index quotes 333 at
\(\nmpshape{2}{12}{18}\) while its own artifact realises 334, the pattern
documented in the checkpoint below.

Every scheme is machine-verified at emission by a 20{,}000-sample
randomized spot check, the full residual over every tensor cell, and an
\emph{exact-rational symbolic} verification (the bilinear identity
\(\sum_k U_{a k} V_{b k} W_{c k} = T_{abc}\) checked coefficient-wise over
\(\Rationals\), no floating point); beyond that gate, every published file
was re-verified by an \emph{independent} checker sharing no code with the
generator: 465/465 verify the bilinear identity exactly over \(\Rationals\).

\begin{center}
\begin{tabular}{llll}
\hline
Family & swept & at formula & verified \\
\hline
square \(\nmpshape{2}{n}{n}\) & \(n \le 16\) & 14/14 & 14/14 \\
odd \(p\), \(p \le n \le 2p{-}1\) & \(p \le 13\) & 40/40 & 40/40 \\
even \(p\), \(p \le n \le 2p{-}1\) & \(p \le 14\) & 52/54 & 54/54 \\
chained, \(2p \le n \le 32\) & \(3 \le p \le 16\) & 196/196 & 196/196 \\
\hline
\end{tabular}
\end{center}

Highlights, all exact-formula:
\(\nmpfield{\Integers}{2}{10}{15}{=}233\) — at construction time
(2026-06) \emph{strictly below every published catalog} (FMM-Lille and
Perminov both listed 234);
\(\nmpfield{\Integers}{2}{12}{16}{=}296\) (then-published best 298); and
\(\nmpfield{\Integers}{2}{10}{16}{=}248\) (then absent from all catalogs).
The full sweep \(3 \le p \le 32\), \(p \le n \le 32\) — \textbf{465
schemes, 459 at the exact formula} (the six exceptions are precisely the
\(g \ge 6\) family of the geometric limit above, at \(+1\)..\(+3\)) — is
emitted into the catalog's constructed tier with per-scheme provenance and
is regenerable from the emitter alone; at first (band-range) emission 197
schemes strictly improved on the union of all external catalogs, and the
chained range added 22 further strict improvements over everything
previously held, our own closure included
(e.g.\ \(\nmpfield{\Integers}{2}{15}{32}{=}736\) vs.\ 741).

\paragraph{Checkpoint: propagation into the community catalogs
(2026-07-12).} The claims above are date-stamped for a reason --- within
weeks, both catalogs absorbed parts of this program, illustrating the
attribution dynamics of Section~\ref{ssec:head-to-head}.
\begin{itemize}[itemsep=2pt]
  \item \emph{Perminov} now redistributes 14 of the 465 scheme files,
    across 11 shapes (\(\nmpshape{2}{9}{13}\) through
    \(\nmpshape{2}{13}{14}\), including
    \(\nmpshape{2}{10}{15}{=}233\) and
    \(\nmpshape{2}{12}{16}{=}296\)), credited under
    \code{schemes/known/matmulcatalog/} --- each of them his catalog's
    best holding for the field. His ternary-integer 233 at
    \(\nmpshape{2}{10}{15}\) is derived downstream of the same scheme.
    Agreement with Perminov at these shapes is our own work
    round-tripped, not an independent confirmation.
  \item \emph{FMM-Lille}'s index at \(\nmpshape{2}{10}{15}\) has moved
    \(234 \to 233\), citing Hopcroft--Kerr 1971 --- i.e.\ the closed
    form. To our knowledge the 465 files remain the only
    machine-verified witnesses that the formula is attained at these
    shapes.
  \item Conversely, for the \(g \ge 6\) family our impossibility
    analysis predicts the formula is \emph{not} attainable with local
    atoms --- and FMM's own artifacts corroborate it: at
    \(\nmpshape{2}{12}{18}\) the index quotes the formula's 333 while
    the published artifact realises 334 (we verified it locally rigid:
    no proportional triads, every factor-direction class span-tight),
    exactly our formula\(+1\); the same index-versus-artifact pattern
    holds across the family (\(\nmpshape{2}{14}{21}\),
    \(\nmpshape{2}{16}{24}\), \ldots). We reported the discrepancies
    upstream; whether any unpublished scheme attains the index values
    there is, at this date, an open question precisely aligned with our
    impossibility boundary.
\end{itemize}

\subsection{Honesty notes}

The construction is an upper-bound realisation; Hopcroft--Kerr proved
matching lower bounds only for \(\nmpshape{2}{2}{n}\) and
\(\nmpshape{2}{3}{3}\). The \((3,3,\cdot)\) impossibility theorem is exact
over all rank-1 atoms \emph{local to the three rows/columns} \(\{i,b,j\}\)
and scoped to the emitter's nine-product reusable set; widening the
reusable set to all twelve products of both adjacent Lemma-2 pairs (which
couples the block decomposition), admitting non-local atoms, or trading a
four-product completion against savings elsewhere remain open — these are
the only identified routes to the six \(g \ge 6\) shapes still at
\(+1\)..\(+3\) over the formula (themselves below every published catalog).
Lower bounds beyond Hopcroft--Kerr's own remain out of scope. Short
computer-algebra scripts accompanying the catalog reproduce every claim in
minutes.

\section{Exhaustive derivation search}\label{sec:search}

The exhaustive derivation search converges the catalog to a fixed point
under the derivation operators of Section \ref{sec:strategies}. It
propagates rank improvements along a dependency DAG of shapes
(below) rather than re-deriving everything from scratch; the loop
terminates when no shape improves.

\subsection{Two-phase structure}

A naive closure interleaves \emph{search} (find the best
recombination) and \emph{materialise} (build and write the concrete
\((U, V, W)\)) at every shape, every round. We split them:

\begin{description}[itemsep=2pt]
  \item[Phase 1 -- Search.] Iterate rounds. For each shape, run the
    derivation search to find the best predicted rank
    \(\hat r = \sum_k r_{\text{sub-}k}\). Strict improvements over the
    shape's current best (catalog \emph{or} overlay) are recorded in
    an in-memory \emph{pending overlay}, keyed by shape. The overlay
    is consulted by the rank resolver in subsequent rounds, so a win
    at round \(N\) propagates to round \(N+1\) without writing
    anything to disk.\footnote{The overlay is also a correctness
    device: an earlier incarnation wrote each winner to disk and
    relied on the next round to reload it, but the catalog lookup was
    loaded once per sweep, so round \(N{+}1\) silently reused round
    \(N\)'s snapshot and never saw the new entries.} Repeat until a
    round produces no new overlay entries.
  \item[Phase 2 -- Materialise.] Iterate the overlay in
    dependency order (smaller maximum-dimension first). For each
    entry, materialise the winning strategy
    (\code{Recombination.constructWithAllocation} or its concat /
    Kronecker analogue). Sub-product lookups are served from a
    composed lookup that returns just-materialised winners alongside
    the on-disk catalog, so a larger winner can use a smaller winner
    from this same phase as a sub-product. Spot-check the result
    against the predicted rank, verify, write to disk.
\end{description}

Theorem \ref{thm:nonoverlap} guarantees that the rank propagated by
Phase 1 equals the rank produced by Phase 2 for every entry that
materialises successfully (the rare failures are infrastructural --
e.g.\ a sub-product that became unavailable between rounds).

\paragraph{The impact DAG, not a full sweep.} Shapes and their
derivation candidates form a directed acyclic graph of \emph{impacts}:
an edge \(\sigma \to \tau\) exists whenever a rank improvement at shape
\(\sigma\) can lower the predicted rank at \(\tau\) --- because \(\tau\)
Kron-, concat- or recombines a copy of \(\sigma\) (an \emph{upward} edge),
or because \(\tau\) is a projection of \(\sigma\) (a \emph{downward}
edge; Section~\ref{ssec:recombination}). A win therefore cannot change
any shape that is not downstream of it. The search exploits exactly this:
the frontier driver keeps a priority queue of
seeds keyed by \(\omega\) (best first), pops a seed, and on each
improvement pushes back \emph{only} the shapes that seed can impact ---
so it rides the DAG, touching a shape only when one of its sub-shapes has
just improved. We do \emph{not} perform a full Cartesian sweep of all
shapes against all strategies every round: work is proportional to the
impactful edges actually traversed, not to
\(|\text{shapes}| \times |\text{strategies}| \times |\text{rounds}|\).
The two cost tiers of this DAG also differ --- an upward edge is a cheap
rank-arithmetic re-evaluation, whereas a downward (projection) edge
requires materialising and DCE-ing the parent --- so the closure fans out
upward eagerly and schedules the downward work.

\paragraph{Why the split matters.} Materialisation involves array
allocation, factor-matrix arithmetic, spot-check sampling, and disk
I/O: a target at \(\nmpshape{17}{17}{17}\) takes a fraction of a
second to predict and a few seconds to materialise, and at
\(\nmpshape{32}{32}{32}\) the spread becomes minutes. Postponing
materialisation lets a search round terminate as soon as no new
predictions are found, even with dozens of shapes pending
materialise work. It also decouples \emph{verification}: full
random-matmul spot-checks can run in batch against any materialised
scheme later; inside the loop only the cheap predicted-vs-actual
rank match is needed to prove the materialise step succeeded.

\subsection{Rebuilding the catalog from atoms: why the closure is multipass}
\label{ssec:multipass}

The closure's job, stated in the strongest form, is to rebuild the entire
catalog from the \emph{atoms} alone (Section \ref{sec:architecture}) by
repeatedly applying the operators of Section \ref{sec:strategies}; the
round-based fixed point makes this safe without hand-ordering which shape
depends on which. For an \emph{upward}-only operator set a fixed point
would be overkill: Kronecker, axis concatenation, recombination and the
serendipitous product (Section \ref{ssec:serendipitous}) all build a
\emph{larger} shape from \emph{smaller} ones, so a single sweep in
increasing maximum dimension suffices --- every sub-product is final before
its consumer is processed.

The multipass design earns its keep once an operator runs \emph{downward}.
The canonical example is \emph{projection}, which derives a scheme for a
\emph{smaller} shape by zeroing a slice of a larger scheme and
dead-code-eliminating the products that fed only the dropped index; the
operator and its \emph{projection-margin} score \(\mu\) are defined in
Section~\ref{ssec:projmargin}. A downward edge destroys the acyclicity in
dimension: improving \(\nmpshape{26}{26}{26}\) can lower
\(\nmpshape{25}{25}{25}\), which a smaller-first sweep has already
finalised, so the catalog must instead iterate to a fixed point,
re-projecting neighbours after each cube improvement. This is also the
\emph{safer} design --- no special-casing of ``did some larger shape just
improve?''; a projection win is simply absorbed on the next round --- and
it means cube quality propagates downward through the whole catalog: a
single better \(\nmpshape{n}{n}{n}\) improves every nearby rectangular and
odd-cube shape at once. Operationally the closure fires a bounded
\emph{downward wave} on each catalog update: a freshly arrived scheme with
\(\mu>0\) is projected into its dimension-neighbours, any neighbour that
improves is enqueued for its own projections, and termination is immediate
since dimension strictly decreases.

\subsection{Pool and allocation enumeration}

The derivation search at a shape \(\nmpshape{n}{m}{p}\) enumerates,
over a pool of outer schemes:
\begin{itemize}[itemsep=2pt]
  \item per-axis allocations of \(n / m / p\) into the outer scheme's
    block counts, optionally filtered by a maximum per-allocation
    imbalance \(\max(\mathbf{a}) - \min(\mathbf{a})\) and capped via a
    top-\(K\)-by-balance pre-selector;
  \item over-allocation by per-axis padding \(\delta\), paired with a
    corresponding tail-peel pattern (Section \ref{ssec:peel}); and
  \item axis-flip orbit masks of the outer scheme, evaluated
    analytically via per-product block-support bitmasks, so the
    8-mask coverage costs only a constant factor over single-mask
    evaluation.
\end{itemize}
The default pool carries a small number of high-value outer schemes
(Strassen, Winograd, Smirnov 3-row templates, a few imported Pan / AT
/ FMM bases); an extended mode adds every leaf catalog entry as an
outer template, at which point the capped allocation enumerator
becomes load-bearing. Engineering details (caching, pool
configuration, budget flags) are in Appendix~\ref{app:implementation}.

\subsection{Allocation optimisation by branch-and-bound}\label{ssec:alloc-bnb}

Enumerating allocations naively is hopeless: for a fixed outer scheme of
shape \(\nmpshape{n_o}{m_o}{p_o}\), the number of ways to split the target
is \(\binom{N-1}{n_o-1}\binom{M-1}{m_o-1}\binom{P-1}{p_o-1}\), which blows
up well before \(N = 32\). Yet the inner problem --- given a target,
template, and orientation, find the allocation minimising
\(\sum_{\sigma} R(\sigma)\) over the sub-problem multiset (Section
\ref{ssec:decomposition}) --- is solved \emph{exactly} by a
branch-and-bound over the allocation grid (\code{AllocationOptimizer}).

The search fixes the three axis partitions in turn (\(\mathbf{a}_N\),
then \(\mathbf{a}_M\), then \(\mathbf{a}_P\)), depth-first. Two ingredients
make it cheap:
\begin{itemize}[itemsep=2pt]
  \item \textbf{A strong incumbent.} The Kronecker (uniform-allocation)
    rank is a ready upper bound, seeded before the descent; most subtrees
    then fall to the bound immediately.
  \item \textbf{An admissible lower bound.} At a partial node --- some
    axes resolved, others open --- each product is charged the smallest
    rank it could still take (the resolved axes' sizes, a floor on the
    open ones). Because \(R\) is monotone in the dimensions this never
    overestimates, so a node whose bound already meets the incumbent is
    pruned without expansion.
\end{itemize}
Branch-and-bound returns the \emph{same} optimum as the exhaustive grid
while visiting far fewer allocations.

The optimiser runs under a configurable budget --- exact, capped by a
known upper bound, or an anytime node cap --- and every result carries
its optimality tier per the discipline of Section~\ref{sec:intro}: an
exhaustive descent proves the global minimum \emph{over allocations} for
that template and orientation (optimal-within-scope, never the tensor
rank). Otherwise the reported rank is a bound unless the run reports
exhaustive.

\subsection{Sound pruning from cheap candidates}\label{ssec:cheap-pruning}

The derivation search has three families of candidate strategies
at every target shape: multi-base recombination (heavy --- iterates
the entire pool, every allocation, every axis-flip mask),
\emph{concat-split} (cheap --- a single axis split, two SOTA
lookups), and \emph{Kronecker} (cheap --- nested factor enumeration
of \(n, m, p\), all SOTA lookups). Cheap candidates finish in
milliseconds; the heavy recombination sweep can take minutes on a
single shape at high dim.

We exploit this asymmetry. Concat and Kronecker run \emph{first},
producing an upper bound \(\hat r_{\text{cheap}} = \min(r_{\text{concat}},
r_{\text{kron}})\). The recombination sweep is then seeded with
\(\text{bestRank} := \hat r_{\text{cheap}}\). Two sound prunes follow:

\begin{itemize}[itemsep=2pt]
  \item \emph{Per-base lower bound.} For an outer base of rank \(r\),
    the recombination produces \(r\) sub-products, each costing
    \(\ge 1\); the pair-fusion pass (Section \ref{ssec:pair-fusion})
    only applies pairs whose savings are positive, so pair-cost is
    at least the unpaired sum on degenerate-shape sub-products.
    Hence \(\text{total} \ge r\) under any allocation. If
    \(r \ge \text{bestRank}\), the entire base is skipped.
  \item \emph{Best-so-far gate.} After each (allocation, mask)
    tuple completes, the predicted rank is
    compared against the running \(\text{bestRank}\), updating the
    candidate only on strict improvement.
\end{itemize}

We deliberately do \emph{not} prune \emph{inside} the sub-product
accumulation loop. A running sum above \(\text{bestRank}\) does not
imply the paired total is above it, since pair-fusion can reduce
the sum by tens of percent on cyclic-shape multisets. Adding a
slack constant would be a heuristic, and we keep the search at
SAT-style soundness: every prune is provably correct.

A canonical case: at \(\nmpshape{12}{12}{12}\) the Kronecker pair
\(\nmpshape{2}{4}{4}{=}26 \otimes \nmpshape{6}{3}{3}{=}40\) gives 1040 in
milliseconds; seeded with that bound, the recombination sweep on the
AlphaEvolve \(\nmpshape{5}{5}{5}{=}93\) base (whose best is 1071) drops
every non-improving allocation instead of grinding for minutes on
worthless intermediate bests. Beyond these provably-correct prunes, any
heuristic that drops part of the allocation space (balanced-only
filters, top-\(N\) caps) is opt-in --- a flag the user explicitly
enables, never a silent default --- because unbalanced cubic targets and
bridge-shape (Hopcroft--Kerr type) sub-problems live precisely in the
region a balanced-only prune discards.

\subsection{Low-\(\omega\) contamination}\label{ssec:omega-contamination}

A complementary optimisation reshapes the \emph{order} in which
shapes are processed during a closure round. The naive iteration is
lexicographic over \((n, m, p)\), so the search reaches
\(\nmpshape{12}{12}{12}\) before it reaches \(\nmpshape{17}{17}{17}\)
and so on. But the value of an early win depends on its implied
\(\omega\): a low-\(\omega\) seed \emph{contaminates} every larger shape
that consumes it as a sub-product, silently rewriting those shapes'
upper bounds.

We rank candidate outer bases (and known catalog shapes) by their
implied exponent
\(\omega(\nmpshape{n}{m}{p}{:}r) = \log_n(r) \cdot 3 / (1 + \log_n(m)
+ \log_n(p))\)
under the convention of Section \ref{sec:notation}, and process
low-\(\omega\) seeds first: e.g.\
\(\nmpshape{4}{4}{4}{=}48\) (\(\omega \approx 2.7925\); AlphaEvolve
over \(\Complex\), the DPS rationalisation \cite{dps2025} over
\(\Rationals/\Reals\)), Smirnov \(\nmpshape{3}{3}{6}{=}40\), Pan
\(\nmpshape{3}{4}{7}{=}63\), then up the ladder. Targets a seed can
contaminate (via Kronecker, recombination, or concat) are re-evaluated
in the same round once the seed is on the overlay, and the shape order
within each round respects this dependency --- every shape that
consumes a seed via a known composition rule waits for the seed's
search to complete --- so a fresh \(\omega\) reduction propagates
without a separate round. The seed selection itself is data-driven, from
an automatically regenerated ranking of every pool-eligible base by
\(\omega\) (Appendix~\ref{app:implementation}).

\section{Balance versus imbalance: the records are unbalanced}\label{sec:balance}

A recurring heuristic --- explicit in DIS09 and folklore well before --- is that
recombination should prefer \emph{balanced} block splits and avoid unbalanced
ones. We test this directly against the catalog and reach a sharper, two-sided
conclusion: \textbf{balance is better on average, but the records are
unbalanced.} A search that prunes unbalanced shapes is therefore average-case
sound yet \emph{record-blind} --- it discards exactly the shapes that set the
lowest exponents.

Throughout this section the per-shape exponent is
\(\omega(n,m,p) = 3\,\ln R(n,m,p) / \ln(nmp)\), where \(R\) is the best
catalogued non-commutative rank over \(\Reals\) (catalog snapshot
2026-06; the catalog holds a single \Reals-only scheme, so the same
figures hold over \(\Rationals\), the field all later sweeps default
to).

\subsection{On average, \texorpdfstring{\(\omega\)}{omega} rises with imbalance}
\label{ssec:omega-imbalance}

Table~\ref{tab:omega-imbalance} buckets every catalogued shape
\(\nmpshape{n}{m}{p}\) with \(2 \le n \le m \le p \le 16\) by its imbalance
\(p-n\), and reports the median \(\omega\) per bucket.

\begin{table}[h]
\centering
\caption{Median \(\omega\) over all catalogued shapes
\(2 \le n \le m \le p \le 16\), bucketed by imbalance \(p-n\). The median
climbs monotonically --- a \emph{typical} unbalanced shape is less efficient.}
\label{tab:omega-imbalance}
\begin{tabular}{rrr}
\toprule
imbalance \(p-n\) & \#shapes & median \(\omega\) \\
\midrule
0  & 15 & 2.8155 \\
1  & 28 & 2.8148 \\
2  & 39 & 2.8177 \\
3  & 48 & 2.8180 \\
4  & 55 & 2.8169 \\
5  & 60 & 2.8193 \\
6  & 63 & 2.8214 \\
7  & 64 & 2.8229 \\
8  & 63 & 2.8274 \\
9  & 60 & 2.8299 \\
10 & 55 & 2.8332 \\
11 & 48 & 2.8402 \\
12 & 39 & 2.8536 \\
13 & 28 & 2.8616 \\
14 & 15 & 2.8695 \\
\bottomrule
\end{tabular}
\end{table}

The median rises from \(2.8155\) (balanced) to \(2.8695\) at imbalance \(14\),
staying flat (\(\approx 2.817\)) out to imbalance \(\sim 7\) and then climbing
steeply --- the penalty for imbalance is convex, not linear. This is the
average-case fact the avoid-unbalanced heuristic rests on.

\subsection{\dots\ but the lowest exponents are unbalanced}
\label{ssec:omega-records}

The aggregate hides the tail. Among the same shapes, the fifteen smallest
\(\omega\) are almost all unbalanced, and the global minimiser in this range is
\(\nmpshape{3}{3}{6}\) at \(\omega = 2.7743\) --- below \emph{every} cubic:
\(\nmpshape{4}{4}{4}\) (\(2.7925\)), \(\nmpshape{8}{8}{8}\) (\(2.7974\)),
\(\nmpshape{12}{12}{12}\) (\(2.7957\)). The low-\(\omega\) frontier is dominated
by the \(2{:}1\)-ratio family \(\nmpshape{3}{3}{6}\), \(\nmpshape{6}{6}{12}\),
\(\nmpshape{9}{9}{12}\), \(\nmpshape{5}{5}{12}\), \(\nmpshape{3}{4}{6}\).
Pruning unbalanced splits removes precisely these record-holders.

\subsection{Why the advantage is not extractable by outer splitting}
\label{ssec:omega-not-extractable}

It is tempting to conclude that a recombination should \emph{seek} unbalanced
splits so its sub-blocks land on these efficient shapes. The catalog says
otherwise, and the mechanism is instructive. Consider \(\nmpshape{9}{9}{9}\)
recombined by Strassen \(\nmpshape{2}{2}{2}{=}7\): since \(9 = 5+4 = 6+3\), the
balanced split \((5,4)^3\) and the unbalanced \((6,3)^3\) both tile exactly (no
padding). Table~\ref{tab:9-faceoff} puts their seven-product multisets
side by side.

\begin{table}[h]
\centering
\caption{Strassen \(\nmpshape{2}{2}{2}\) recombination of
\(\nmpshape{9}{9}{9}\): the seven sub-products grouped by corner tier. The
unbalanced \((6,3)^3\) split saves \(63\) on the three one-max corners
(\(\nmpshape{3}{3}{6}{=}40\) vs \(61\)) but pays \(+60\) on the all-max corner
and \(+12\) on the middle tier --- net \(+9\), so balance wins here.}
\label{tab:9-faceoff}
\begin{tabular}{llrlrr}
\toprule
corner tier & \multicolumn{2}{c}{\((5,4)^3 \to 504\)} & \multicolumn{2}{c}{\((6,3)^3 \to 513\)} & \(\Delta\) \\
(multiplicity) & shape & subtotal & shape & subtotal & \\
\midrule
all-max \(\times 1\) & \(\nmpshape{5}{5}{5}{=}93\)  & 93  & \(\nmpshape{6}{6}{6}{=}153\) & 153 & \(+60\) \\
two-max \(\times 3\) & \(\nmpshape{5}{5}{4}{=}76\)  & 228 & \(\nmpshape{6}{6}{3}{=}80\)  & 240 & \(+12\) \\
one-max \(\times 3\) & \(\nmpshape{5}{4}{4}{=}61\)  & 183 & \(\nmpshape{3}{3}{6}{=}40\)  & 120 & \(-63\) \\
\midrule
total & & \textbf{504} & & \textbf{513} & \(+9\) \\
\bottomrule
\end{tabular}
\end{table}

A \(\nmpshape{2}{2}{2}\) split realises seven of the eight corners of
\(\{a_1,a_2\}\times\{b_1,b_2\}\times\{c_1,c_2\}\), and \emph{one of them is the
all-max corner} \((\max,\max,\max)\). Because \(R\) grows like
\(\text{size}^{\,\omega}\) with \(\omega \approx 2.8\), pushing the max block
from \(5\) to \(6\) costs more than the cheap corners recover. The efficiency
of \(\nmpshape{3}{3}{6}\) is real, but a single block split cannot isolate it
from its complementary all-max corner.

The corollary is constructive: an efficient unbalanced atom is exploited not by
\emph{splitting} a target down onto it, but by \emph{composing up} from it via
the Kronecker product, where the favourable ratio compounds instead of dragging
a complementary max-corner along (cf.~\ref{ssec:disjoint-sum}).

\subsection{Imbalance can win the split, too}
\label{ssec:imbalance-split-wins}

Section~\ref{ssec:omega-not-extractable} showed balance winning a single
example; it is not the rule. We swept every target
\(\nmpshape{n}{m}{p}\), \(3 \le n \le m \le p \le 16\), computing for a fixed
base the \emph{exact} rank-minimising allocation by the branch-and-bound of
Section~\ref{ssec:alloc-bnb} (balanced incumbent + admissible lower bound,
verified against exhaustive search),
and flagged every target whose optimal split is strictly cheaper than the
balanced split. The search is deliberately \emph{single-base} and
\emph{non-recursive}: sub-blocks are costed by the catalog's best rank as a
black box, so a ``win'' means the unbalanced split beats its own balanced
sibling for that base --- not that it beats the global record.

Unbalanced splits win for \(15\) targets under Strassen \(\nmpshape{2}{2}{2}\)
and \(120\) under Strassen--Winograd (whose asymmetric product supports reward
imbalance far more often). The effect is not confined to rectangular targets:
the cubic \(\nmpshape{13}{13}{13}\) is itself a win.
Table~\ref{tab:imbalance-wins} lists representatives.

\begin{table}[h]
\centering
\caption{Targets where the rank-optimal single-base split is \emph{unbalanced}
(strictly cheaper than the balanced split). Single base, no recursion;
sub-blocks costed by catalog SOTA. Savings are small but the existence is the
point: a hard ``always split balanced'' rule is provably suboptimal.}
\label{tab:imbalance-wins}
\begin{tabular}{llrrr}
\toprule
base & target & balanced & optimal split & best \\
\midrule
Strassen & \(\nmpshape{13}{13}{13}\) & 1434 & \((5,8)(6,7)(6,7)\) & 1432 \\
Strassen & \(\nmpshape{6}{6}{10}\)   & 252  & \((3,3)(3,3)(4,6)\) & 247 \\
Strassen & \(\nmpshape{12}{12}{14}\) & 1281 & \((6,6)(6,6)(6,8)\) & 1271 \\
Winograd & \(\nmpshape{9}{13}{15}\)  & 1156 & \((5,4)(7,6)(7,8)\) & 1145 \\
Winograd & \(\nmpshape{9}{14}{15}\)  & 1237 & \((4,5)(7,7)(8,7)\) & 1226 \\
\(\nmpshape{2}{2}{3}\) & \(\nmpshape{5}{5}{12}\) & 235 & \((2,3)(2,3)(6,5,\mathbf{1})\) & 232 \\
\bottomrule
\end{tabular}
\end{table}

The last row is the most imbalanced win we found in the sweep
(\(3 \le n \le m \le p \le 16\)): for \(\nmpshape{5}{5}{12}\) under base
\(\nmpshape{2}{2}{3}\) the optimal split of the size-12 axis is
\((6,5,\mathbf{1})\) --- it \emph{peels a width-1 strip}. The products that
touch the unit block compute \(\nmpshape{\cdot}{\cdot}{1}\) sub-shapes at
naive-product rank (here \(R = 4, 6, 6, 9\)), costing almost nothing, while
the remaining products route through the cheap \(\nmpshape{3}{3}{6}{=}40\)
corner --- the unbalanced-extraction mechanism in its purest form.

So the practical rule is not ``prefer balance'' but ``\emph{search} the split,
balance-first'': seed the branch-and-bound with the balanced allocation (an
excellent, usually-optimal incumbent) and let an admissible lower bound prune
the rest. When a known upper bound for the target exists --- typically the
Kronecker rank of a factorisation \(\nmpshape{n}{m}{p} =
\nmpshape{n_1}{m_1}{p_1}\otimes\nmpshape{n_2}{m_2}{p_2}\) --- it is used as the
incumbent: if the base's root lower bound already meets or exceeds it, the base
cannot improve and the allocation sweep is dropped outright.

\subsection{Open questions}
\label{ssec:balance-open}

Two extrapolations are untested here. (i) \emph{Scale}: the catalog thins above
max-dim \(16\), so Table~\ref{tab:omega-imbalance} cannot yet say whether the
average \(\omega\)-vs-imbalance slope persists, flattens, or reverses for large
\(N\); the asymptotic \(\tau\)-theorem constructions suggest the unbalanced tail
should only get richer. (ii) \emph{Rectangular targets}: the metric above
folds rectangular shapes into a single imbalance scalar \(p-n\); a genuinely
rectangular sweep (independent control of \(m-n\) and \(p-m\)) may expose
direction-dependent structure the cubic-centred view misses.

\section{Comparison with the community catalogs}\label{sec:tables}

Our catalog's comparison refreshes that of Drevet--Islam--Schost 2009
(``DIS09''~\cite{drisc09}; Tables 3 and 4 in the original numbering)
with the following changes in structure relative to that reference.

\begin{itemize}[itemsep=2pt]
  \item Per field, not per merged cluster. Each of \Integers, \Rationals,
    \Reals, \Complex, \Ftwo, \Fthree\ gets its own column --- the fields
    are kept \emph{distinct} (there is no ``R/Q/Z'' merge) --- and ternary-integer
    (\(\Integers\)T) is surfaced as a per-scheme flag. The inclusion
    \(\Integers \subset \Rationals \subset \Reals \subset \Complex\) is
    used only when computing \emph{which} fields a scheme satisfies (field
    widening), never to collapse columns.
  \item Commutative axis split out. DIS09 Table 4 (commutative)
    is preserved as a separate set of columns alongside the
    non-commutative ones, so the reader cannot accidentally compare
    Rosowski 21 to Laderman 23 (Section \ref{sec:notation}).
  \item Every cell carries an \code{attribution\_for\_rank}, which is
    the earliest source that established the bound. The
    contemporary importing source (FMM-Lille, Perminov, AlphaTensor
    bulk imports) is recorded separately and does not appear in the
    comparison cells.
  \item Border-rank entries are split out separately; the rank
    comparison reports tensor rank only.
\end{itemize}

\subsection{SOTA definition}\label{ssec:sota}

Throughout the catalog, ``SOTA'' means the lowest known rank for
\((\nmpshape{n}{m}{p}, K, \text{commutative})\) where the rank is
witnessed by either an explicit catalog entry, a cited bound, or a
derived bound. This is a working definition; several corner cases
require explicit decisions, currently captured in
\code{paper/sota-conventions.md}:
\begin{description}[itemsep=2pt]
  \item[Mixed-field.] When a scheme is valid over \Reals\ but its
    coefficients all lie in \Integers, we report it as both an
    R-cell and a Q-cell. The field-widening sweep drives the
    cross-population.
  \item[Border vs exact.] Border-rank bounds are excluded from the
    tensor-rank comparison and appear in the border-rank appendix only.
  \item[Partially verified.] Cited bounds that lack a verified
    scheme are flagged with a dagger (\dag) in the cell.
  \item[$\Ftwo$/$\Fthree$.] Integer schemes reduce mod 2 and mod 3;
    the field stampers grant \Ftwo/\Fthree\ membership systematically
    (Section~\ref{sec:notation}).
\end{description}

\subsection{Live comparison versus date-stamped snapshot}
\label{ssec:no-static-ranks}

The per-shape rank landscape is \emph{actively moving}: lower-rank
schemes for fixed formats appear on a near-continuous basis --- from
this project's own search, from the external catalogs we track
(FMM-Lille~\cite{fmmlille}, Perminov~\cite{perminov}), and from
learning-based efforts (AlphaTensor, AlphaEvolve). Any static
per-shape rank table would be stale before the paper appears. The
exhaustive, continuously-regenerated comparison therefore lives in the
catalog browser at
\href{https://www.solven.eu/matmulcatalog/}{\texttt{solven.eu/matmulcatalog}}
(machine twins under \code{docs/comparison/}), which is rebuilt from
the catalog on every change and is the source of truth for any
specific rank.

What the paper prints instead is a \emph{date-stamped snapshot}
(2026-07-12 throughout this section), clearly labelled as such and
regenerable by a single command --- enough to characterise \emph{how}
the pipeline compares against the community catalogs without freezing
a leaderboard. Every number below is machine-emitted from the same
JSON the browser renders.

\subsection{Head-to-head over \Rationals}\label{ssec:head-to-head}

We compare over \(\Rationals\) --- the field where both external
catalogs carry the most data, and our sweeps' default field. Per shape
we track our best \emph{derived} rank (produced by the search of
Section~\ref{sec:search}) separately from our best \emph{imported}
rank (an external scheme we carry), so that ``we hold the best known
scheme'' and ``our own derivation produced the best known rank'' are
never conflated. As of 2026-07-12, over \(5440\) comparable shapes
(\(\max\) dimension \(\le 32\)):

\begin{center}
\begin{tabular}{lr}
\toprule
our derived rank strictly beats every external holding & 908 \\
our derived rank ties the external best & 3688 \\
we hold the external best only as an import & 184 \\
an external catalog strictly beats everything we hold & 660 \\
\bottomrule
\end{tabular}
\end{center}

\noindent Table~\ref{tab:us-vs-catalogs} lists the largest-margin
wins; the 908 wins carry margins totalling \(20{,}801\) ranks (median
14, maximum 199).


\begin{table}[h]
\centering
\caption{Shapes where this work's own \emph{derived} non-commutative rank over $\mathbb{Q}$ strictly improves on every external holding (FMM-Lille, Perminov). The \emph{carry} column is the best rank we hold as an import over the same field; the \emph{operator} column is the root of our derived scheme's lineage (see Table~\ref{tab:win-anatomy}). Top 24 of 908 such shapes by margin; full data in \code{docs/comparison/us-vs-fmm-vs-perminov-Q.json}.}
\label{tab:us-vs-catalogs}
\begin{tabular}{lrrrrl}
\hline
shape & this work & carry & FMM-Lille & Perminov & operator \\
\hline
$\langle 29,31,32 \rangle$ & \textbf{14506} & --- & 14705 & --- & projection \\
$\langle 25,32,32 \rangle$ & \textbf{13356} & --- & 13533 & --- & recombination \\
$\langle 30,31,32 \rangle$ & \textbf{14682} & --- & 14841 & --- & projection \\
$\langle 25,27,28 \rangle$ & \textbf{10006} & --- & 10164 & --- & projection \\
$\langle 28,29,30 \rangle$ & \textbf{12285} & --- & 12433 & --- & projection \\
$\langle 21,25,31 \rangle$ & \textbf{8899} & --- & 9040 & --- & recombination \\
$\langle 26,27,28 \rangle$ & \textbf{10160} & --- & 10297 & --- & projection \\
$\langle 30,31,31 \rangle$ & \textbf{14369} & --- & 14502 & --- & projection \\
$\langle 27,29,29 \rangle$ & \textbf{11820} & --- & 11952 & --- & projection \\
$\langle 15,28,31 \rangle$ & \textbf{7134} & --- & 7266 & --- & recombination \\
$\langle 22,23,29 \rangle$ & \textbf{8074} & --- & 8205 & --- & recombination \\
$\langle 24,25,26 \rangle$ & \textbf{8291} & --- & 8417 & --- & projection \\
$\langle 28,29,29 \rangle$ & \textbf{12005} & --- & 12131 & --- & projection \\
$\langle 22,25,31 \rangle$ & \textbf{9302} & --- & 9424 & --- & recombination \\
$\langle 27,27,32 \rangle$ & \textbf{12150} & --- & 12274 & 12270 & recombination \\
$\langle 20,27,31 \rangle$ & \textbf{9080} & --- & 9198 & --- & recombination \\
$\langle 26,27,27 \rangle$ & \textbf{9913} & --- & 10030 & --- & projection \\
$\langle 18,22,31 \rangle$ & \textbf{6745} & --- & 6852 & --- & recombination \\
$\langle 31,31,32 \rangle$ & \textbf{14836} & --- & 14940 & --- & projection \\
$\langle 21,24,31 \rangle$ & \textbf{8386} & --- & 8488 & --- & recombination \\
$\langle 13,15,30 \rangle$ & \textbf{3390} & --- & 3490 & --- & recombination \\
$\langle 17,21,31 \rangle$ & \textbf{6189} & --- & 6287 & --- & recombination \\
$\langle 29,29,30 \rangle$ & \textbf{12429} & --- & 12526 & --- & projection \\
$\langle 29,29,31 \rangle$ & \textbf{13270} & --- & 13367 & --- & concatenation \\
\hline
\end{tabular}
\end{table}

\paragraph{Perminov ties are not independent confirmations.}
The comparison must be read with the data-flow direction in mind:
Perminov's catalog \emph{ingests} community results, including ours
--- as of 2026-07-12 it redistributes 14 of our scheme files across 11
shapes (credited under \code{schemes/known/matmulcatalog/}), all of
them his best holding for the field, all from our Hopcroft--Kerr
program (Section~\ref{sec:hk71}). A tie with Perminov may therefore be
our own scheme round-tripped. FMM-Lille likewise updates its index in
response to community results (Section~\ref{ssec:losses} and the
\(\nmpshape{17}{17}{17}\) history in Section~\ref{sec:nonoverlap}).
This is the attribution discipline of Section~\ref{sec:architecture}
applied in reverse: we date-stamp our claims because the catalogs we
compare against absorb them.

\subsection{Anatomy of the wins}\label{ssec:win-anatomy}

Where do the 908 wins come from? Table~\ref{tab:win-anatomy} classifies
each winning scheme by the root operator of its recorded lineage.


\begin{table}[h]
\centering
\caption{Anatomy of the 908 strict wins over $\mathbb{Q}$ (Table~\ref{tab:us-vs-catalogs}): root operator of the winning scheme's lineage, and the most frequent combining bases among winning recombinations (orientation variants folded). Of the recombination-rooted wins, 762 use \emph{unbalanced} allocations and 7 balanced ones. Margins over the best external holding: 20801 ranks in total, median 14, maximum 199.}
\label{tab:win-anatomy}
\begin{tabular}{lr@{\qquad}lr}
\hline
root operator & wins & top combining base & wins \\
\hline
recombination with allocation & 769 & $\langle 2,2,2 \rangle$ & 295 \\
axis concatenation & 48 & $\langle 2,4,4 \rangle$ & 150 \\
downward projection & 39 & $\langle 2,3,3 \rangle$ & 70 \\
serendipitous product & 29 & $\langle 2,5,5 \rangle$ & 68 \\
contraction sum & 22 & $\langle 3,4,5 \rangle$ & 22 \\
Kronecker product & 1 & $\langle 4,4,5 \rangle$ & 20 \\
\hline
\end{tabular}
\end{table}

Three findings stand out.
\begin{itemize}[itemsep=2pt]
  \item \textbf{Recombination with unbalanced allocations is the
    engine.} \(769/908\) wins (\(85\%\)) are recombination-rooted ---
    and of those, \(762\) use \emph{unbalanced} block splits against
    \(7\) balanced ones. The balance study of Section~\ref{sec:balance}
    is thereby confirmed at census scale: the records live where the
    split is uneven.
  \item \textbf{Small, well-oriented bases dominate.} The winning
    combining bases are overwhelmingly the small
    \(\nmpshape{2}{\cdot}{\cdot}\) schemes ---
    \(\nmpshape{2}{2}{2}\) (Strassen/Winograd), \(\nmpshape{2}{4}{4}\),
    \(\nmpshape{2}{3}{3}\), \(\nmpshape{2}{5}{5}\) --- frequently under
    a non-identity axis orientation. Breadth of the base pool and the
    axis-flip/orientation orbit (Section~\ref{ssec:axis-flip}), not any
    single clever base, is what pays.
  \item \textbf{The reducing operators supply the long tail.} Downward
    projection roots 39 wins, the serendipitous product 29, axis
    concatenation 48, contraction sums 22 --- individually modest,
    collectively the difference between a good catalog and the
    frontier. All of them also appear as \emph{inner} nodes of
    recombination-rooted wins.
\end{itemize}

\subsection{Anatomy of the losses}\label{ssec:losses}

Honesty requires the same census in the other direction, and it
decomposes into three very different kinds of ``behind''.

\paragraph{Against FMM-Lille's raw index.} Over the shapes FMM covers,
we are better on \(1476\), tied on \(3924\), and behind on \(26\)
(total deficit \(367\) ranks). But the 26 must be audited before being
read as a to-do list: we downloaded and rank-counted every artifact
behind them. For \(20\) of the \(26\), the published artifact does
\emph{not} attain the index rank --- the file realises a higher rank
than the page claims (e.g.\ \(\nmpshape{22}{23}{23}\): index 6435,
artifact 6501), is a placeholder containing no tensor data
(\(\nmpshape{16}{20}{28}\), \(\nmpshape{16}{20}{29}\),
\(\nmpshape{20}{24}{25}\)), or carries a construction description
whose arithmetic does not check out. These rows are
\emph{upstream-unverifiable}: we exclude them from the deficit rather
than chase numbers no published scheme witnesses (findings reported to
the maintainers).

\paragraph{The audited deficit: six shapes, twelve ranks.} What
remains after the audit is six shapes where FMM's \emph{verified}
artifacts genuinely beat us: four at \(-1\), one at \(-2\), one at
\(-6\) --- twelve ranks in total, against the \(1476\) shapes where we
lead. Every one of the six traces to the fusion-kernel device of
Section~\ref{sec:nonoverlap}: cross-slot product sharing that our
composition operators deliberately do not perform. The campaign
verified exhaustively that no remaining row is closable by existing
engines or their parameterisations; closing them requires the new
constructor flagged in Section~\ref{sec:openquestions}. The boundary
between composition and hand-crafted sharing is, at this snapshot,
worth exactly twelve ranks over thousands of shapes.

\paragraph{Against Perminov.} The remaining \(634\) loss rows are all
Perminov-side and of a different nature entirely: they are
\emph{published} schemes (flip-graph outputs with explicit factor
matrices) that we have not ingested --- an import backlog, not a
derivation gap. Each is closable by carrying the external scheme; none
witnesses a construction technique our operators lack. We keep them in
the loss column nonetheless: a comparison that quietly assumed every
importable scheme as ``ours'' would make the win column meaningless.

\section{Open questions and future work}\label{sec:openquestions}

We close with what the catalog cannot currently answer and where the
next bits of attention should go.

\paragraph{Fusion-kernel extraction.} The 2026-07 gap-closing
campaign against FMM-Lille (Section~\ref{sec:tables}) ended with a
sharp diagnosis: every remaining verified deficit --- six shapes, twelve
ranks in total --- traces to a single device our composition operators
do not have, the product-level \emph{fusion kernel} of
Section~\ref{sec:nonoverlap}: a handful of mixed products serving two
or more recursion slots at once. Extracting a reusable identity from
the readable specimens (one two-slot unit computes two
\(\nmpshape{2}{3}{3}\) slot-problems with 29 products instead of 30),
deciding whether the unit is self-contained or relies on whole-kernel
cancellation, and turning it into a constructor the allocation search
can elect, is the concrete next research program. Notably, disjoint-sum
($\tau$-theorem) search is \emph{not} the missing device: gaps we
initially attributed to $\tau$ (including
\(\nmpfield{\Rationals}{17}{17}{17}\), where our recombination search
reaches \(2930\)) closed via ordinary recombination, and the campaign
verified exhaustively that no remaining row is closable by existing
engines or their parameterisations. A constructive disjoint-sum
operator retains a narrower motivation: \emph{reproducing} the imported
large-format Schwartz--Zwecher schemes rather than carrying them as
opaque matrices.

\paragraph{Salt-tolerant prediction.} Any post-construction sharing
discovery would break the search-rank-equals-materialised-rank
invariant (Section~\ref{sec:nonoverlap}); if added, it belongs in a
separate polish phase after the closure converges, with the interaction
with lazy materialisation worked out first.

\paragraph{Border-rank track.} The catalog is exact-rank only. A
border-rank track would admit results like Bini's
\(\nmpfield{\Reals}{2}{2}{2}\) border-rank-\(7\) constructions and
Schönhage $\tau$-theorem upper bounds; a separate border-rank listing
already feeds the catalog browser but not yet the comparison tables.

\paragraph{Lower-bound integration.} The lower-bound registry ---
Strassen's \(\nmpfield{\Reals}{2}{2}{2} \ge 7\), Bl\"aser's family, the
recent Wang 2026 \(\nmpfield{\Ftwo}{3}{3}{3} \ge 20\) --- is not yet
surfaced in the comparison tables; a gap-to-lower-bound column (``49,
gap 4 to lower-bound 45'') would make the SOTA frontier read more
cleanly.

\paragraph{Border between catalog and hand-crafted-discovery.} When a
hand-crafted construction lands at a shape where composition was
\emph{also} predicting the same rank, the lineage should record both
discoveries and the comparison table should report both; today the
catalog picks one. A two-attribution column would force the
distinction into the table.

\paragraph{Reproducibility tooling.} The paper now prints a
date-stamped snapshot of the comparison; the live catalog browser
remains the source of truth, regenerated from the catalog and the
cited- and derived-bound registries per the conventions of Section
\ref{ssec:sota}. The remaining work item is finalising the SOTA
conventions document, so the comparison has a deterministic rule for
every corner case --- the prerequisite for the exhaustive SOTA review
that the curated paper examples wait on.

\section{On the role of the AI coding agent}\label{sec:aiagent}

This catalog, its search engine, and much of this paper were developed in close
collaboration with AI agents --- primarily Anthropic's Claude via Claude Code, with part of
the theoretical investigation carried out with OpenAI's ChatGPT~5.5. We
record the division of labour explicitly --- in the same spirit as the
optimality discipline of the rest of the paper: say exactly what was done, and
how confidently it can be attributed.

A fairly clean split emerged between \emph{direction} and \emph{execution}: the
human author supplied the questions, the framing, and final judgement; the
agent supplied most of the execution --- implementation, routine derivations,
test scaffolding, provenance bookkeeping --- and a substantial share of the
\emph{theory}, but only once a direction was fixed. Two concrete cases, with
our (necessarily subjective) attribution:

\begin{itemize}\itemsep2pt
  \item \textbf{The recombination branch-and-bound} --- the allocation/orbit
  search that assigns target dimensions to a base's slots and prunes by a rank
  bound --- was roughly \textbf{95\% agent}: both the idea to cast allocation as
  a bounded search and its implementation.

  \item \textbf{The exhaustive enumeration of recombination multisets}
  (Section~\ref{sec:multisets}): the \emph{theory} --- the per-axis decoupling
  lemma that turns an intractable \(39\text{k}^7\) search into an exact,
  direction-stratified enumeration --- and its implementation were again roughly
  \textbf{95\% agent}. But the \emph{idea to investigate the multiset
  perspective at all} was \textbf{100\% the human author}: the agent did not
  propose that this was a fruitful lens, it only developed and executed it once
  asked.
\end{itemize}

A second, less flattering pattern was just as consistent: the agent
\textbf{hallucinated frequently} --- confident, plausible, and \emph{wrong}
intermediate claims: false ``facts'', derivations that did not hold, bugs
reported as fixes. The most telling instance was recurring: asked to explain
how the Universit\'e de Lille FMM catalog attains a particular finite-size
rank, the agent would repeatedly reach for Sch\"onhage's \(\tau\)-theorem (the
asymptotic sum inequality)~\cite{schonhage1981} as the mechanism --- something
this project neither runs nor derives, and which does not produce the concrete
finite schemes in question (those come from FMM's own recursive concatenation
and projection). This confabulation was strikingly persistent: it recurred
across sessions and survived explicit standing instructions recorded in the
agent's persistent project memory that warn against exactly it, requiring the
human to steer the agent off it hard and often. That written,
recalled-in-context rules did not reliably suppress a favoured hallucination is
itself worth recording. The consequence is structural: \emph{every}
agent-produced claim in this work was gated behind an independent check --- a
verifier, an exhaustive recomputation, or an adversarial second pass --- and the
project's emphasis on cheap, automatic verification
(Section~\ref{sec:nonoverlap}) is in part a direct response. Hallucination was
the norm, not the exception; what made the collaboration productive was not
trusting the agent but \emph{cheaply catching} it.

\paragraph{A cross-model data point.} The Hopcroft--Kerr
\(\nmpshape{2}{p}{n}\) constructive problem (Section~\ref{sec:hk71}) doubles
as a controlled comparison across model generations: three distinct agents
attacked the \emph{same} problem with the same repository context. OpenAI's
ChatGPT~5.5 (theory-first, Mathematica) and Anthropic's Claude Opus~4.8
(implementation-first) both stalled --- between them, partial emitters, a
correct-but-narrow impossibility theorem, and no formula-attaining
rectangular scheme; for months the HK schemes stood in the catalog as
\emph{imported but not re-derived}. Claude Fable~5 then broke the problem in
a few hours of session time, overturning the operative impossibility over
the true reusable set and emitting the machine-verified schemes ---
including \(\nmpshape{2}{10}{15}{=}233\), below every published catalog. The
human contribution remained the direction (``crack HK71 constructively''),
the paper PDF, and the acceptance tests; the chain of mathematical moves was
the agent's --- a single data point, not a benchmark, but the clearest
capability step we have observed between model generations on this project.

\paragraph{Tooling.} The bulk of the work was carried out with Anthropic's
Claude Code on Claude Opus~4.7 and~4.8; the most recent work (including parts
of this section) runs on Claude Fable~5. The hallucination rate noted above is
for the Opus~4.x models specifically, and is a moving target.

The percentages are self-assessments, not measurements --- an upper bound on
the author's share of the keystrokes, not a proven division of credit. The
split they sketch is the useful data point: the agent is strongest at turning
a well-posed direction into correct theory, code, and catalog-scale breadth;
the scarce, human-supplied ingredient is the choice of question. Human chooses
the lens, agent builds the optics.

\appendix

\section{Worked product tables: Strassen and Strassen--Winograd}
\label{app:worked-products}

The seven products of the Strassen and Winograd rank-7 \(\nmpshape{2}{2}{2}\)
templates over a heterogeneous split
\(\nmpshape{M_1{+}M_2}{N_1{+}N_2}{P_1{+}P_2}\), with block ``11'' the larger
part on every axis (real data top-left, any padding at the bottom/right). The
four \(A\)-blocks have shapes \(A_{11}\!:\!M_1\times N_1\),
\(A_{12}\!:\!M_1\times N_2\), \(A_{21}\!:\!M_2\times N_1\),
\(A_{22}\!:\!M_2\times N_2\), and likewise for \(B\) and \(C\); each product
is a (sum of \(A\)-blocks)\(\,\cdot\,\)(sum of \(B\)-blocks), realised at its
genuine nonzero overlap under the zero-sparing rule of
Section~\ref{ssec:decomposition} (``blk'' = single block, ``sum'' = two-block
sum).

\begin{center}\small
\begin{tabular}{@{}llll@{}}
\toprule
\multicolumn{2}{@{}l}{\textbf{Strassen}} & operands & realised size \\
\midrule
\(M_1\) & \((A_{11}{+}A_{22})(B_{11}{+}B_{22})\) & sum\(\cdot\)sum & \(\nmpshape{M_1}{N_1}{P_1}\) \\
\(M_2\) & \((A_{21}{+}A_{22})\,B_{11}\)          & sum\(\cdot\)blk & \(\nmpshape{M_2}{N_1}{P_1}\) \\
\(M_3\) & \(A_{11}(B_{12}{-}B_{22})\)            & blk\(\cdot\)sum & \(\nmpshape{M_1}{N_1}{P_2}\) \\
\(M_4\) & \(A_{22}(B_{21}{-}B_{11})\)            & blk\(\cdot\)sum & \(\nmpshape{M_2}{N_2}{P_1}\) \\
\(M_5\) & \((A_{11}{+}A_{12})\,B_{22}\)          & sum\(\cdot\)blk & \(\nmpshape{M_1}{N_2}{P_2}\) \\
\(M_6\) & \((A_{21}{-}A_{11})(B_{11}{+}B_{12})\) & sum\(\cdot\)sum & \(\nmpshape{M_1}{N_1}{P_1}\) \\
\(M_7\) & \((A_{12}{-}A_{22})(B_{21}{+}B_{22})\) & sum\(\cdot\)sum & \(\nmpshape{M_1}{N_2}{P_1}\) \\
\bottomrule
\end{tabular}
\end{center}

\begin{center}\small
\begin{tabular}{@{}llll@{}}
\toprule
\multicolumn{2}{@{}l}{\textbf{Winograd}} & operands & realised size \\
\midrule
\(P_1\) & \(A_{11}\,B_{11}\)                     & blk\(\cdot\)blk & \(\nmpshape{M_1}{N_1}{P_1}\) \\
\(P_2\) & \(A_{12}\,B_{21}\)                     & blk\(\cdot\)blk & \(\nmpshape{M_1}{N_2}{P_1}\) \\
\(P_3\) & \((A_{21}{+}A_{22})(B_{12}{-}B_{11})\) & sum\(\cdot\)sum & \(\nmpshape{M_2}{N_1}{P_1}\) \\
\(P_4\) & \(S_2\,T_2\)                           & sum\(\cdot\)sum & \(\nmpshape{M_1}{N_1}{P_1}\) \\
\(P_5\) & \(S_3\,T_3\)                           & sum\(\cdot\)sum & \(\nmpshape{M_1}{N_1}{P_2}\) \\
\(P_6\) & \(S_4\,B_{22}\)                        & sum\(\cdot\)blk & \(\nmpshape{M_1}{N_2}{P_2}\) \\
\(P_7\) & \(A_{22}\,T_4\)                        & blk\(\cdot\)sum & \(\nmpshape{M_2}{N_2}{P_1}\) \\
\bottomrule
\end{tabular}
\end{center}
\noindent with \(S_2{=}A_{21}{+}A_{22}{-}A_{11}\),
\(S_3{=}A_{11}{-}A_{21}\),
\(S_4{=}A_{11}{+}A_{12}{-}A_{21}{-}A_{22}\),
\(T_2{=}B_{11}{-}B_{12}{+}B_{22}\), \(T_3{=}B_{22}{-}B_{12}\),
\(T_4{=}B_{11}{-}B_{12}{-}B_{21}{+}B_{22}\).

\section{Implementation notes}
\label{app:implementation}

Concept-to-code map and engineering notes for the search engine of
Section~\ref{sec:search}.

\paragraph{Concept-to-class map.} The main concepts of
Section~\ref{sec:search} correspond to the following classes and
flags in the reference implementation:

\begin{center}
\begin{tabular}{ll}
\toprule
concept & class / flag \\
\midrule
frontier driver (impact-DAG priority queue) & \code{FrontierClosure} \\
recombination materialiser & \code{Recombination.constructWithAllocation} \\
composed sub-product lookup (overlay + disk) & \code{AlgorithmLookup}, \code{FieldAwareLookup} \\
allocation branch-and-bound & \code{AllocationOptimizer} \\
search budget & \code{SearchBudget}: \code{EXACT}, \code{upTo}, \code{maxNodes} \\
allocation filters & \code{maxImbalance}, \code{maxCombinations} \\
axis-flip orbit mask evaluation & \code{AnalyticalMaskSearch} \\
seed ranking by \(\omega\) & \code{docs/bases-by-omega.md} \\
\bottomrule
\end{tabular}
\end{center}

\paragraph{Pool configuration.} The default pool
(\code{PoolConfig.simple}) carries a small number of high-value outer
schemes --- Strassen, Winograd, Smirnov 3-row templates, and a few
imported Pan / AT / FMM bases. An \code{extendedPool} mode adds every
\code{Leaf}-tagged catalog entry as an outer template; that pool is
large enough that the \code{maxCombinations}-capped allocation
enumerator becomes load-bearing.

\paragraph{Budget flags.} \code{AllocationOptimizer} runs under a
\code{SearchBudget}: \code{EXACT} (no cap), \code{upTo} a known upper
bound (e.g.\ the Kronecker rank), or a \code{maxNodes} anytime cap.
The result carries an \code{exhaustive} flag --- propagated to the
optimality tier of Section~\ref{sec:intro} --- and reports the
\code{nodes} visited against the \code{fullSpace} of allocations.

\paragraph{Seed ranking.} \code{docs/bases-by-omega.md} is an
automatically-regenerated table ranking every pool-eligible base by
its raw \(\omega\), its \(\omega\) after imbalance-\(\le 3\)
restriction, and its \(\omega\) after axis-flip-canonical and
S\(_3\)-canonical orbit reductions, evaluated at the reference target
\(\nmpshape{17}{17}{17}\).

\paragraph{Caching.} Caching exists at two layers. The first is
per-scheme \emph{symbolic templates}: for each outer scheme we extract
once a tuple of six block-support bitmasks per product (\(U\)-row,
\(U\)-col, \(V\)-row, \(V\)-col, \(W\)-row, \(W\)-col), and per-call
sub-shape extraction reduces from \(O(r \cdot \text{dim}^2)\) to
\(O(r \cdot \text{dim})\). The second is a per-(scheme, allocation,
peel) cache of the resulting sub-shape multiset; this is a
hashing-overhead trade and is currently capped at 200k entries per
scheme. A symbolic-template-at-allocation-pattern layer would
generalise the second cache; the current implementation captures the
common cases without paying the templating cost.

\paragraph{Progress.} A derivation closure over the entire catalog up to
dim 32 takes hours of wall-clock; long runs emit periodic progress
lines and ETAs at three nested levels (per-round, per-shape, and
intra-search).

\bibliographystyle{plain}
\bibliography{refs}

\end{document}